\documentclass[letterpaper]{article} 
\usepackage{aaai24}  
\usepackage{times}  
\usepackage{helvet}  
\usepackage{courier}  
\usepackage[hyphens]{url}  
\usepackage{graphicx} 
\usepackage{natbib}  
\usepackage{caption} 
\usepackage{algorithm}
\usepackage{algorithmic}

\usepackage{newfloat}
\usepackage{listings}
\DeclareCaptionStyle{ruled}{labelfont=normalfont,labelsep=colon,strut=off} 
\floatstyle{ruled}
\newfloat{listing}{tb}{lst}{}
\floatname{listing}{Listing}
\title{General Performance Evaluation for Competitive Resource Allocation Games via Unseen Payoff Estimation}
\author{
    N'yoma Diamond\textsuperscript{\rm 1,\rm 2}, Fabricio Murai\textsuperscript{\rm 1}
}
\affiliations{
    \textsuperscript{\rm 1}Worcester Polytechnic Institute\\
    \textsuperscript{\rm 2}University of Cambridge
}

\nocopyright

\usepackage{xspace}

\usepackage[american]{babel}

\usepackage{amsfonts}
\usepackage{amsmath}
\usepackage{amsthm}
\usepackage{amssymb}
\usepackage{mathtools}

\usepackage{dsfont}

\usepackage[capitalize]{cleveref}

\usepackage{booktabs}
\usepackage{cellspace}
\usepackage{arydshln}
\usepackage{multirow}

\usepackage{dblfloatfix}

\usepackage{thmtools}
\usepackage{thm-restate}

\usepackage{subcaption}
\usepackage{tikz}
\usetikzlibrary{positioning, calc, decorations.pathreplacing}
\pgfdeclarelayer{edgelayer}
\pgfdeclarelayer{nodelayer}
\pgfsetlayers{edgelayer,nodelayer,main}

\newcommand{\CB}{\ensuremath{\mathcal{CB}}\xspace}

\newcommand{\ttsc}[1]{\texttt{\textsc{#1}}}

\newcommand{\expectation}{\mathop{\mathbb{E}}\limits}

\newcommand{\twodots}{\mathinner {\ldotp \ldotp}}

\makeatletter
\newcommand\footnoteref[1]{\protected@xdef\@thefnmark{\ref{#1}}\@footnotemark}
\makeatother

\newtheorem{lemma}{Lemma}
\newtheorem{prop}{Proposition}

\crefname{prop}{Proposition}{Propositions}

\crefname{section}{\S}{\S\S}

\begin{document}

\maketitle

\begin{abstract}
Many high-stakes decision-making problems, such as those found within cybersecurity and economics, can be modeled as competitive resource allocation games. In these games, multiple players must allocate limited resources to overcome their opponent(s), while minimizing any induced individual losses. However, existing means of assessing the performance of resource allocation algorithms are highly disparate and problem-dependent. 
As a result, evaluating such algorithms is unreliable or impossible in many contexts and applications, especially when considering differing levels of feedback.
To resolve this problem, we propose a generalized definition of payoff which uses an arbitrary user-provided function. This unifies performance evaluation under all contexts and levels of feedback. Using this definition, we develop metrics for evaluating player performance, and estimators to approximate them under uncertainty (i.e., bandit or semi-bandit feedback). These metrics and their respective estimators provide a problem-agnostic means to contextualize and evaluate algorithm performance. To validate the accuracy of our estimator, we explore the Colonel Blotto (\CB) game as an example. To this end, we propose a graph-pruning approach to efficiently identify feasible opponent decisions, which are used in computing our estimation metrics. Using various resource allocation algorithms and game parameters, a suite of \CB games are simulated and used to compute and evaluate the quality of our estimates. These simulations empirically show our approach to be highly accurate at estimating the metrics associated with the unseen outcomes of an opponent's latent behavior. 
\end{abstract}

\section{Introduction}

The necessity for strategic resource allocation pervades many critical real-world domains, such as cybersecurity, economics, and epidemiology. We oftentimes have limited resources for our goals, and thus need to allocate them carefully to reach an optimal outcome. For example, the failure or success of a cyberattack may depend on the allocations of offensive or defensive resources across a variety of digital systems. If the defending entity fails to allocate enough resources to defend a vulnerable system, the attacking entity may succeed in an attack on that system. Yet, overcompensating by allocating excess resources may result in deficiencies elsewhere. In practice, accurately identifying optimal ways to efficiently allocate resources is a challenging task, especially when we lack information about the adversary's allocations---i.e., under bandit or semi-bandit feedback.

Competitive resource allocation games emerge as a natural way to model these problems. One such highly explored game is the Colonel Blotto (\CB) game. In the \CB game, two players compete by allocating limited resources to a number of battlefields with the goal of overpowering their opponent on as many of them as possible. Under this interpretation, the received payoff directly relates to the player's allocations for the round, henceforth referred to as their ``\textbf{decision}''. Substantial work has been done in the past to create algorithms that approach approximately optimal strategies for the \CB game and other analogous resource allocation games when given only minimal information to learn from. However, in order to develop and evaluate these algorithms, varying and inconsistent assumptions are made to narrow the problem space. This makes them highly scenario-specific, and thus incomparable and potentially inapplicable to many real-world problems. In particular, we failed to identify any extant literature which can be applied to the context of mutually adaptive adversaries; the broader context such that both agents may be actively attempting to actively hinder their opponent in some way. This circumstance is applicable to many real-world scenarios (such as the cybersecurity example mentioned earlier), and can be used to generalize the evaluation of resource allocation algorithms more broadly.

To remedy this gap, we develop a general definition of payoff for competitive resource allocation games with the goal of unifying performance evaluation independent of context and feedback. Furthermore, we propose a suite of practical evaluation metrics based on our general payoff definition, and practical means for approximating them under uncertainty. Focusing on the \CB game, we propose a graph-pruning approach for identifying feasible opponent decisions, which can be used to efficiently compute our estimated payoff metrics. Finally, to improve the efficiency and accuracy of our approach, we prove a number of bounds for feasible opponent allocations under semi-bandit feedback. This provides an approach through which future research and development may assess the performance of resource allocation algorithms under theoretical contexts and active real-world use. To empirically validate the effectiveness of our approach, we perform simulations of \CB games and evaluate them utilizing our proposed metrics and algorithm. Our experiments show that our approach is highly effective at estimating the true payoff metrics associated with the actual opponent behavior, even under uncertainty.

\textit{Outline.} In \cref{sec:related work} we discuss the related work on competitive resource allocation games and performance evaluation. In \cref{sec:general payoff} we propose our definition of generalized payoff. In \cref{sec:payoff metrics,sec:estimating metrics}, we propose, respectively, our evaluation metrics utilizing general payoff and their corresponding estimators. In \cref{sec:problem formulation} we present the \CB game used as a case study for our metrics. In \cref{sec:modeling with graphs} we propose a novel technique for narrowing down the set of feasible opponent decisions in the \CB game. In \cref{sec:estimating opponent allocations} we identify bounds on feasible opponent allocations under semi-bandit feedback. In \cref{sec:empirical analysis} we utilize the proposed techniques to simulate and evaluate the quality of our estimation metrics with respect to their true counterparts. In \cref{sec:conclusion} we summarize the conclusions of our work.

\section{Related Work}\label{sec:related work}

Many researchers explore the problem of strategic resource allocation through the lens of combinatorial bandits. This is a variation on the highly-explored multi-armed bandit problem, such that the agent may simultaneously pull multiple arms within some budget, as opposed to only a single arm. This problem introduces the distinction between ``bandit'' and ``semi-bandit'' feedback \cite{audibertRegretOnlineCombinatorial2014}. Under bandit feedback, the agent receives a single aggregate payoff as feedback, which makes it difficult to determine the individual impact of each component of a decision. Conversely, under semi-bandit feedback scenario, the agent receives distinct feedback about the payoffs associated with each component of a decision. This provides more information than bandit feedback, while still having uncertainty about the true behavior of how rewards are provided. 

Under combinatorial bandits, \citet{zuoCombinatorialMultiarmedBandits2021} propose two online algorithms using combinatorial decision spaces for the discrete and continuous resource cases. They consider a general reward function, making their algorithms applicable to many contexts. However, their algorithms depend on the usage of an unspecified oracle function which may not be practical or possible to implement efficiently, and thus compare against. Additionally, \citet{kocakEfficientLearningImplicit2014} propose a version of the existing \ttsc{Exp3} \cite{auerNonstochasticMultiarmedBandit2012} algorithm leveraging observability graphs to consider the potential value of unexplored decisions.

\citet{vuCombinatorialBanditsSequential2019,vuPathPlanningProblems2020} reframe the combinatorial interpretation as a path-planning problem with inspiration from the observability graphs from \citet{kocakEfficientLearningImplicit2014}. To do so, they expand on the existing \ttsc{ComBand} \cite{cesa-bianchiCombinatorialBandits2012} and \ttsc{Exp3} \cite{auerNonstochasticMultiarmedBandit2012} algorithms. \citet{vuModelsSolutionsStrategic2020} presents a comprehensive survey and analysis of the existing literature regarding combinatorial bandits and the \CB game to propose a suite of resource allocation algorithms utilizing varying levels of feedback. To the best of our knowledge, the works of \citet{vuCombinatorialBanditsSequential2019,vuPathPlanningProblems2020,vuModelsSolutionsStrategic2020} are the only present literature explicitly studying general algorithms for mutually adversarial resource allocation games. In this work we provide metrics and techniques which may be used to reliably evaluate and compare algorithms such as these.

\section{General Payoff}\label{sec:general payoff}

The payoff received by a player $p$ in the set of players $P$ is predicated on their chosen decision $\pi$ and their opponent's decision $\phi$. Thus, we denote the set of all feasible decisions (e.g., valid allocations in a round of the \CB game) available to player $p$ and their opponent as $\Pi_p$ and $\Phi_p$, respectively. Using this notation, we describe general payoff to be a function $~{L_p: \Pi_p \times \Phi_p \mapsto \mathbb{Q}}$ that returns the scalar payoff awarded to player $p$ if they play decision $\pi\in\Pi_p$ and their opponent plays decision $\phi\in\Phi_p$. Thus for a given round $t$, player $p$'s payoff for that round is calculated as
\begin{equation}
    L_p^t = L_p(\pi^t,\phi^t),
\end{equation}
where $\pi^t\in\Pi_p$ denotes the decision played by player $p$ in round $t$ and $\phi^t\in\Phi_p$ denotes the decision played by their opponent. Further, by fixing $\phi^t$ but allowing $\pi$ to vary, we introduce the following shorthand:
\begin{equation}
    L_p^t(\pi) \coloneqq L_p(\pi, \phi^t),
\end{equation}
which represents the payoff that player $p$ would receive for playing a specific decision $\pi$ in round $t$.

Note that a generalized interpretation of regret can be produced trivially by taking the difference of the general payoff between any two possible games. That is, for a comparison of two arbitrary player decisions against one arbitrary opponent decision, generalized regret may be calculated as $L_p(\pi,\phi)-L_p(\pi',\phi)$; as $L_p^t(\pi)-L_p^t(\pi')$ given a fixed opponent decision for round $t$; or as $L_p^t(\pi) - L_p^t$ when also given a fixed player decision for round $t$.

\section{Payoff Metrics}\label{sec:payoff metrics}

Using our generalized definition of payoff, we propose two useful metrics: Max Payoff and Expected Payoff. \textbf{Max Payoff} is defined as the maximum possible (i.e., optimal) payoff that can be received by player $p$ against a particular decision by opponent $p'$. It can be formally expressed as the function
\begin{equation}
    L_p^*(\phi) \coloneqq \max_{\pi\in\Pi_p} L_p(\pi,\phi).
\end{equation}
For a fixed opponent decision $\phi^t$, we denote it as:
\begin{equation}
    L_p^{*t} \coloneqq \max_{\pi\in\Pi_p} L_p^t(\pi).
\end{equation}
It is important to note that $L_p^{*t}$ represents the maximum possible payoff that player $p$ can achieve for a given round $t$. Thus this metric is highly valuable towards computing regret and any associated metrics.

Let the player's decision in a given round $t$ be a random variable $\pi \sim \mathcal{D}^t(\Pi_p)$, where $\mathcal{D}^t(\Pi_p)$ is the probability distribution over the player's decisions in $\Pi_p$ in round $t$. $\mathcal{D}^t$ is parameterized by $t$ because the distribution may change depending on the game's history (such as when the players are adaptive adversaries). Using this distribution, we define the \textbf{Expected Payoff} as player $p$'s expectation of payoff over $\mathcal{D}^t(\Pi_p)$. Formally, it is given by
\begin{equation}
    \widehat{L}_p(\phi) \coloneqq \expectation_{\pi\sim \mathcal{D}^t(\Pi_p)}\left[L_p(\pi,\phi)\right].
\end{equation}
Over a large number of plays against the same opponent decision, the player's average payoff will approach the Expected Payoff. Therefore this metric can be used as a reference point to identify whether an algorithm is behaving as intended or is outperforming the expected (average) performance of another algorithm.
Yet again, we can leverage our previously described shorthand by fixing the opponent decision $\phi^t$:
\begin{equation}
    \widehat{L}_p^t \coloneqq \expectation_{\pi\sim \mathcal{D}^t(\Pi_p)}\left[L_p^t(\pi)\right].
\end{equation}

We note that the necessary analysis to identify $\mathcal{D}^t(\Pi_p)$ is outside the scope of this paper. Therefore, for the purpose of experimentation, we addres this using the Uniform Decision Assumption described in \cref{sec:estimating metrics}.

\section{Estimating Metrics Under Uncertainty}\label{sec:estimating metrics}

The key challenge in computing the proposed metrics is that we often do not observe $\phi$. Thus, to approximate them under uncertainty (i.e., bandit or semi-bandit feedback), we propose three associated estimation metrics: Observable Max Payoff, Supremum Payoff, and Observable Expected Payoff. These metrics are agnostic to whether the player receives bandit or semi-bandit feedback. This is done by identifying and using a set of feasible opponent decisions that would result in the observed game. That is, given a player $p$, a round $t$, a decision $\pi^t\in\Pi_p$, and some round-specific feedback (such as an observed payoff $L_p^t$ or a vector of payoffs associated with a decision), the player computes a set of feasible opponent decisions $\Phi_p^t$. Given a fixed player decision $\pi^t$, any and all decisions $\phi\in\Phi_p^t$ must produce the same feedback as that observed for round $t$. Assuming the user has sufficient knowledge of the nature of the game, it should always be possible to identify feasible opponent decisions in some capacity. A specific approach using the \CB game as an example is discussed in \cref{sec:modeling with graphs,sec:estimating opponent allocations}.

We propose two approximations of Max Payoff, each with distinct purposes: Observable Max Payoff and Supremum Payoff. We define \textbf{Observable Max Payoff} as the expectation of Max Payoff over $\mathcal{D}^t(\Phi_p^t)$: 
\begin{equation}
    L_p^{*t} \approx \expectation_{\phi \sim \mathcal{D}^t(\Phi_p^t)} \left[\max_{\pi\in\Pi_p} L_p(\pi,\phi)\right].
\end{equation}
Observable Max Payoff may be used as a direct approximation of the true value of Max Payoff. This is because over a large number of rounds the running average of Max Payoff for a particular opponent decision should approach the Observable Max Payoff, since the latter is the expectation of the former (assuming $\mathcal{D}^t(\Phi_p^t)$ does not change significantly). 

We also specify the \textbf{Supremum Payoff} as the minimum possible Max Payoff over $\Phi_p^t$: 
\begin{equation}
    \sup(L_p^t) \coloneqq \min_{\phi\in\Phi_p^t} \left[\max_{\pi\in\Pi_p} L_p(\pi,\phi)\right].
\end{equation}
Supremum Payoff represents the minimal upper bound on possible payoff. In other words, it the Supremum Payoff is the best possible payoff that player $p$ can receive if their opponent played their best possible decision. This is the worst-case scenario for $p$, where the opponent's decision minimizes $p$'s best possible payoff. Additionally, Supremum Payoff represents the maximum payoff the player can guarantee is achievable in a particular round, given the information available to them. Hence, it may be valuable to use Supremum Payoff as an objective function for player $p$. Furthermore, Supremum Payoff is guaranteed to be equivalent to or underestimate the true Max Payoff, making it a pessimistic metric. However, in situations where the opponent has substantially more resources than the player, the value of Supremum Payoff may approach the observed payoff $L_p^t$, as the opponent will have many more possible decisions which favor them. That is, it is likely that there exists an opponent decision that makes it impossible for the player to improve their payoff. One of the notable benefits of Supremum Payoff over Observable Max Payoff is that it does not rely on knowing the nature of $\mathcal{D}^t(\Phi_p^t)$.


Similar to Observable Max Payoff, in order to approximate Expected Payoff, we compute its expectation over $\Phi_p^t$. We refer to this estimation metric as \textbf{Observable Expected Payoff}. This is equivalent to the expectation of payoff with respect to both $\mathcal{D}^t(\Pi_p)$ and $\mathcal{D}^t(\Phi_p^t)$:
\begin{equation}
    \widehat{L}_p^t \approx \expectation_{\phi \sim \mathcal{D}^t(\Phi_p^t),\pi \sim \mathcal{D}^t(\Pi_p)}\left[L_p(\pi,\phi)\right].
\end{equation}

\subsection{The Uniform Decision Assumption}\label{sec:UDA}

In practice, it is often functionally impossible to identify the true nature of $\mathcal{D}^t(\Phi_p)$ without knowing the opponent's algorithm and its parameters. Hence, to compute our estimation metrics, we introduce the Uniform Decision Assumption (\textbf{UDA}). Under the UDA, we assume that the true distribution of opponent decisions $\mathcal{D}^t(\Phi_p)$ is approximately equivalent to the uniform distribution. 
Clearly, adversarial algorithms do not produce decisions uniformly. However, we believe the UDA to be an adequate means to enable meaningful approximation of opponent behavior when lacking a method of identifying $\mathcal{D}^t(\Phi_p)$. We implicitly validate the effects of this assumption in our experimental analysis of our estimation metrics, as any significant error resultant from using the UDA should cause the quality of any estimates dependent on it to be highly inaccurate. Specifically, the UDA is used when computing Observable Max Payoff and Observable Expected Payoff due to requiring knowledge of $\mathcal{D}^t(\Phi_p)$ (and $\mathcal{D}^t(\Pi_p)$ in the case of Observable Expected Payoff).

\section{Formulation of the \CB Game Example}\label{sec:problem formulation}

To explore the usage and efficacy of our metrics and estimates, we consider the example of the \CB game. An instance of the \CB game is defined as a two-player repeated constant-sum game of (potentially unknown) length $T\in\mathbb{N}_1$ such that the set of players is denoted by $P=\{A, B\}$. Player $p\in P$ and their opponent $p'$ are allotted some fixed number of resources $N_p,N_{p'}\in\mathbb{N}_1$. Let $K\in\mathbb{N}_1$ be the number of battlefields, such that each battlefield can be represented by an integer $i\in I =[1 \twodots K]$. For each round $t\in[1 \twodots T]$, the resource allocation by player $p$ to a battlefield $i\in I$ is denoted $\pi^t_i$, while their opponent's allocation is denoted $\phi^t_i$. The sum of allocations by any player $p$ in a given round $t$ are fixed, such that $\sum_{i=1}^K \pi^t_i = N_p$. We specify that all allocations are discrete, such that $\pi^t_i, \phi^t_i\in \mathbb{N}_0$. For each round all of the players' resources are renewed and a static one-shot \CB game is played in which each player produces a decision (i.e., vector of allocations) denoted $\pi^t$ for player $p$ and $\phi^t$ for $p'$, such that
\begin{align}
    \pi^t &\coloneqq \langle \pi^t_1, \pi^t_2, \ldots, \pi^t_K \rangle, \\
    \phi^t &\coloneqq \langle \phi^t_1, \phi^t_2, \ldots, \phi^t_K \rangle.
\end{align}

To enforce the constant-sum nature of the game, we enforce a bias when both players allocate the same number of resources to a battlefield (i.e., a draw). Without loss of generality, we specify the variable $\delta_p\in\{0,1\}$ indicating whether player $p$ wins ($1$) or loses ($0$) draws. Thus for a given battlefield $i$ and round $t$, if $\pi^t_i + \delta_p > \phi^t_i$ then player $p$ wins the battlefield. Otherwise, their opponent $p'$ wins the battlefield. Thus, the payoff function used with the \CB game is 
\begin{equation}
    L_p(\pi,\phi) \coloneqq \sum_{i\in I}[\pi_i + \delta_p > \phi_i].
\end{equation}

Given a player $p$, for each battlefield $i$ in round $t$, we denote the associated payoff as $\ell^{t,i}_p$ such that $\ell^{t,i}_p = 1$ if $p$ won the battlefield, or $\ell^{t,i}_p = 0$ if they lost. The vector of payoffs $\mathcal{L}_p^t$ awarded to a player $p$ in a given round is denoted as
\begin{equation}
    \mathcal{L}_p^t \coloneqq \langle \ell^{t,1}_p, \ell^{t,2}_p, \ldots, \ell^{t,K}_p \rangle,
\end{equation}
and the total scalar payoff received by player $p$ at round $t$ is
\begin{equation}
    L_p^t = L_p(\pi^t,\phi^t) = \sum_{i=1}^K \ell^{t,i}.
\end{equation}
Note that computing Max Payoff for the \CB game is trivial as explained in \cref{supp:max payoff alg}.\footnote{\label{foot:supplementary material}See supplementary materials.}

Under bandit feedback, player $p$ only receives their total payoff $L_p^t$, while under semi-bandit feedback, they receive the vector payoff $\mathcal{L}_p^t$. For the purposes of our exploration in this paper, we choose to focus on semi-bandit feedback to provide greater ability to narrow down the set of feasible opponent decisions via bounding the possible allocations, making our estimates more accurate and easier to compute.

This formulation is a slight modification of that used by \citet{robersonColonelBlottoGame2006} and \citet{schwartzHeterogeneousColonelBlotto2014}. Our alterations are as follows: Firstly, we do not assume that the player that loses draws must have the same or fewer number of resources compared to their opponent. Secondly, we allow the player to provided a vector containing the respective payoffs for each battlefield, instead of the aggregate score (i.e., semi-bandit instead of bandit feedback).

\section{Modeling Feasible \CB Games}\label{sec:modeling with graphs}

We adapt the path-planning approach from \citet{vuCombinatorialBanditsSequential2019} to model feasible opponent decisions within the \CB game. This model utilizes a graph-based interpretation of the \CB to efficiently explore varying decisions and their associated outcomes. Notably, this approach accommodates varying levels of feedback received by the player (i.e., bandit or semi-bandit). Using this model, we develop an efficient technique for generating the set of decisions available to a player or their opponent in each round of the \CB game.

\subsection{The \CB Decision Graph}

For a \CB game with $K$ battlefields indexed by $I=[1\twodots K]$, we focus on a given player $p$ with resources $N_p$. \citet{vuCombinatorialBanditsSequential2019} have shown that there exists a DAG $G_{K,N_p}$, which we henceforth refer to as a ``decision graph,'' such that the set of all decisions $\Pi_p$ playable by $p$ maps one-to-one against the set of all paths through $G_{K,N_p}$ (see \cref{fig:graph_example}). 

Let $\mathcal{V}$ be the set of all vertices in the graph. Each vertex has two coordinates $i$ and $n$ representing its position, such that $v_{i,n}\in\mathcal{V}$ is a vertex located at coordinate $(i,n)$. Values of $i$ and $n$ are discrete, such that $i \in \{0\} \cup I$ and $n \in [0\twodots N_p]$. Thus $i$ represents a particular battlefield in the set of battlefields $I$, plus an additional $i=0$ position, and $n$ represents the cumulative amount of resources allocated on the path up to and including a given vertex. Vertices are present at every point in space $[1 \twodots K-1] \times [0 \twodots N_p]$, in addition to the vertices $s \coloneqq v_{0,0}$ and $d \coloneqq v_{K,N_p}$. That is,
\begin{equation}
    \mathcal{V} \coloneqq \{v_{i, n}\ |\ i\in I\setminus\{K\} \wedge n\in[0\twodots N_p] \} \cup \{s,d\}.
\end{equation}

Every vertex $v_{i,n}$ such that $i>0$ (i.e., where $v_{i,n} \neq s$) has inward edges from all vertices $v_{i-1,n'}$ where $n'\in [0 \twodots n]$. 
For any directed edge connecting a vertex $v_{i-1,n'}$ to vertex $v_{i,n}$, denoted $v_{i-1,n'} \to v_{i,n}$, we assign a weight of $n-n'$. In doing so, the weight observed when moving along any edge represents the player's allocation $\pi_i$ (i.e., $\pi_i = n-n'$) to battlefield $i$. Thus the set of all edges $\mathcal{E}$ is
\begin{equation}
    \mathcal{E} \coloneqq \{v_{i-1,n'} \to v_{i,n}\ |\ v_{i-1,n'},v_{i,n} \in \mathcal{V} \wedge n' \leq n \}.
\end{equation}

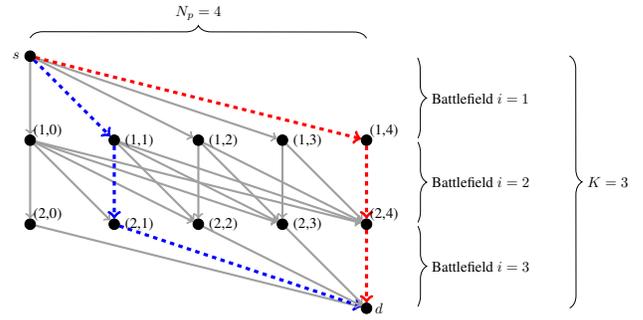
\begin{figure}[t]
\centering
\resizebox{\linewidth}{!}{
\begin{tikzpicture}[
shorten <>/.style = {shorten >=#1mm, shorten <=#1mm},
brc/.style = {decorate, decoration={brace, amplitude=3mm}},
vertex/.style = {fill=black, draw=black, shape=circle, minimum size=2.5mm, inner sep=0mm},
edge/.style = {draw={rgb,255: red,160; green,160; blue,160}, ->, line width=0.5mm},
path1/.style = {->, dashed, line width=0.75mm, draw=blue},
path2/.style = {->, dashed, line width=0.75mm, draw=red}
]
	\begin{pgfonlayer}{nodelayer}
		\draw (0, 6) node [vertex, label=left:{$s$}] (0) {};
	
        \draw (0, 4) node [vertex, label={[label distance=-2mm]above right:{(1,0)}}] (1)  {};
		\draw (2, 4) node [vertex, label=right:{(1,1)}] (2)  {};
		\draw (4, 4) node [vertex, label=right:{(1,2)}] (3)  {};
		\draw (6, 4) node [vertex, label=right:{(1,3)}] (4)  {};
        \draw (8, 4) node [vertex, label={[label distance=-2mm]above right:{(1,4)}}]       (5)  {};
  
		\draw (0, 2) node [vertex, label={[label distance=-2mm]above right:{(2,0)}}] (6)  {};
		\draw (2, 2) node [vertex, label=right:{(2,1)}] (7)  {};
		\draw (4, 2) node [vertex, label=right:{(2,2)}] (8)  {};
		\draw (6, 2) node [vertex, label=right:{(2,3)}] (9)  {};
        \draw (8, 2) node [vertex, label={[label distance=-2mm]above right:{(2,4)}}] (10) {};
  
		\draw (8, 0) node [vertex, label={[label distance=-0.5mm]right:{$d$}}] (11) {};

        \node (12) at (8, 6) {};
	\end{pgfonlayer}
	\begin{pgfonlayer}{edgelayer}
		\draw [edge] (0) -- (1);
		\draw [path1] (0) -- (2);
		\draw [edge] (0) -- (3);
		\draw [edge] (0) -- (4);
		\draw [path2] (0) -- (5);
  
		\draw [edge] (1) -- (6);
		\draw [edge] (1) -- (7);
		\draw [edge] (1) -- (8);
		\draw [edge] (1) -- (9);
		\draw [edge] (1) -- (10);

        \draw [path1] (2) -- (7);
		\draw [edge] (2) -- (8);
		\draw [edge] (2) -- (9);
		\draw [edge] (2) -- (10);

        \draw [edge] (3) -- (8);
		\draw [edge] (3) -- (9);
		\draw [edge] (3) -- (10);

        \draw [edge] (4) -- (9);
		\draw [edge] (4) -- (10);

        \draw [path2] (5) -- (10);

        \draw [edge] (6) -- (11);
        \draw [path1] (7) -- (11);
        \draw [edge] (8) -- (11);
        \draw [edge] (9) -- (11);
        \draw [path2] (10) -- (11);

        \draw[brc, shorten <>=1, decoration={raise=1cm}] (12.east) -- (5.east) node [midway, right=1.3cm] {Battlefield $i=1$};
        \draw[brc, shorten <>=1, decoration={raise=1cm}] (5.east) -- (10.east) node [midway, right=1.3cm] {Battlefield $i=2$};
        \draw[brc, shorten <>=1, decoration={raise=1cm}] (10.east) -- (11.east) node [midway, right=1.3cm] {Battlefield $i=3$};
        \draw[brc, decoration={raise=4.7cm}] (12.east) -- (11.east) node [midway, right=5cm] {$K=3$};
        \draw[brc, decoration={raise=0.3cm}] (0.north) -- (12.north) node [midway, above=0.6cm] {$N_p=4$};
	\end{pgfonlayer}
\end{tikzpicture}
}
\caption{Decision graph $G_{3,4}$ for a game with $K=3$ battlefields given $N_p=4$ resources. Blue path represents the decision $\langle1,0,3\rangle$; red path represents the decision $\langle4,0,0\rangle$.}
\label{fig:graph_example}
\end{figure}

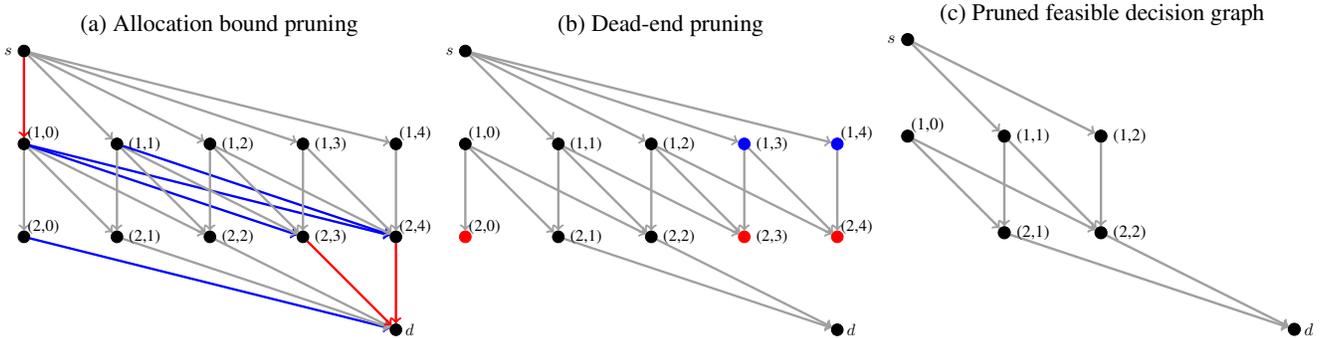
\begin{figure*}[b]
\centering
\begin{subfigure}{0.33\linewidth}
\caption{Allocation bound pruning}
\label{subfig:bound pruning}
\resizebox{\linewidth}{!}{%
\begin{tikzpicture}[
shorten <>/.style = {shorten >=#1mm, shorten <=#1mm},
brc/.style = {decorate, decoration={brace, amplitude=3mm}},
vertex/.style = {fill=black, draw=black, shape=circle, minimum size=2.5mm, inner sep=0mm},
edge/.style = {draw={rgb,255: red,160; green,160; blue,160}, ->, line width=0.5mm},
invalidred/.style = {->, line width=0.5mm, draw=red},
invalidblue/.style = {->, line width=0.5mm, draw=blue}
]
	\begin{pgfonlayer}{nodelayer}
		\draw (0, 6) node [vertex, label=left:{$s$}] (0) {};
	
            \draw (0, 4) node [vertex, label={[label distance=-2mm]above right:{(1,0)}}] (1)  {};
		\draw (2, 4) node [vertex, label=right:{(1,1)}] (2)  {};
		\draw (4, 4) node [vertex, label=right:{(1,2)}] (3)  {};
		\draw (6, 4) node [vertex, label=right:{(1,3)}] (4)  {};
            \draw (8, 4) node [vertex, label={[label distance=-2mm]above right:{(1,4)}}]       (5)  {};
  
		\draw (0, 2) node [vertex, label={[label distance=-2mm]above right:{(2,0)}}] (6)  {};
		\draw (2, 2) node [vertex, label=right:{(2,1)}] (7)  {};
		\draw (4, 2) node [vertex, label=right:{(2,2)}] (8)  {};
		\draw (6, 2) node [vertex, label=right:{(2,3)}] (9)  {};
            \draw (8, 2) node [vertex, label={[label distance=-2mm]above right:{(2,4)}}] (10) {};
  
		\draw (8, 0) node [vertex, label={[label distance=-0.5mm]right:{$d$}}] (11) {};
	\end{pgfonlayer}
 
	\begin{pgfonlayer}{edgelayer}
		\draw [invalidred] (0) -- (1);
		\draw [edge] (0) -- (2);
		\draw [edge] (0) -- (3);
		\draw [edge] (0) -- (4);
		\draw [edge] (0) -- (5);
  
		\draw [edge] (1) -- (6);
		\draw [edge] (1) -- (7);
		\draw [edge] (1) -- (8);
		\draw [invalidblue] (1) -- (9);
		\draw [invalidblue] (1) -- (10);

            \draw [edge] (2) -- (7);
		\draw [edge] (2) -- (8);
		\draw [edge] (2) -- (9);
		\draw [invalidblue] (2) -- (10);

            \draw [edge] (3) -- (8);
		\draw [edge] (3) -- (9);
		\draw [edge] (3) -- (10);

            \draw [edge] (4) -- (9);
            \draw [edge] (4) -- (10);

            \draw [edge] (5) -- (10);

            \draw [invalidblue] (6) -- (11);
            \draw [edge] (7) -- (11);
            \draw [edge] (8) -- (11);
            \draw [invalidred] (9) -- (11);
            \draw [invalidred] (10) -- (11);
	\end{pgfonlayer}
\end{tikzpicture}
}
\end{subfigure}%
\begin{subfigure}{0.33\linewidth}
\caption{Dead-end pruning}
\label{subfig:dead end pruning}
\resizebox{\linewidth}{!}{%
\begin{tikzpicture}[
shorten <>/.style = {shorten >=#1mm, shorten <=#1mm},
brc/.style = {decorate, decoration={brace, amplitude=3mm}},
vertex/.style = {fill=black, draw=black, shape=circle, minimum size=2.5mm, inner sep=0mm},
edge/.style = {draw={rgb,255: red,160; green,160; blue,160}, ->, line width=0.5mm},
deadendred/.style = {fill=red, draw=red, shape=circle, minimum size=2.5mm, inner sep=0mm},
deadendblue/.style = {fill=blue, draw=blue, shape=circle, minimum size=2.5mm, inner sep=0mm}
]
	\begin{pgfonlayer}{nodelayer}
		\draw (0, 6) node [vertex, label=left:{$s$}] (0) {};
	
            \draw (0, 4) node [vertex, label={[label distance=-2mm]above right:{(1,0)}}] (1)  {};
		\draw (2, 4) node [vertex, label=right:{(1,1)}] (2)  {};
		\draw (4, 4) node [vertex, label=right:{(1,2)}] (3)  {};
		\draw (6, 4) node [deadendblue, label=right:{(1,3)}] (4)  {};
            \draw (8, 4) node [deadendblue, label={[label distance=-2mm]above right:{(1,4)}}]       (5)  {};
  
		\draw (0, 2) node [deadendred, label={[label distance=-2mm]above right:{(2,0)}}] (6)  {};
		\draw (2, 2) node [vertex, label=right:{(2,1)}] (7)  {};
		\draw (4, 2) node [vertex, label=right:{(2,2)}] (8)  {};
		\draw (6, 2) node [deadendred, label=right:{(2,3)}] (9)  {};
            \draw (8, 2) node [deadendred, label={[label distance=-2mm]above right:{(2,4)}}] (10) {};
  
		\draw (8, 0) node [vertex, label={[label distance=-0.5mm]right:{$d$}}] (11) {};
	\end{pgfonlayer}
 
	\begin{pgfonlayer}{edgelayer}
		\draw [edge] (0) -- (2);
		\draw [edge] (0) -- (3);
		\draw [edge] (0) -- (4);
		\draw [edge] (0) -- (5);
  
		\draw [edge] (1) -- (6);
		\draw [edge] (1) -- (7);
		\draw [edge] (1) -- (8);

            \draw [edge] (2) -- (7);
		\draw [edge] (2) -- (8);
		\draw [edge] (2) -- (9);

            \draw [edge] (3) -- (8);
		\draw [edge] (3) -- (9);
		\draw [edge] (3) -- (10);

            \draw [edge] (4) -- (9);
            \draw [edge] (4) -- (10);

            \draw [edge] (5) -- (10);

            \draw [edge] (7) -- (11);
            \draw [edge] (8) -- (11);
	\end{pgfonlayer}
\end{tikzpicture}
}
\end{subfigure}%
\begin{subfigure}{0.33\linewidth}
\caption{Pruned feasible decision graph}
\label{subfig:final decision graph}
\resizebox{\linewidth}{!}{%
\begin{tikzpicture}[
shorten <>/.style = {shorten >=#1mm, shorten <=#1mm},
brc/.style = {decorate, decoration={brace, amplitude=3mm}},
vertex/.style = {fill=black, draw=black, shape=circle, minimum size=2.5mm, inner sep=0mm},
edge/.style = {draw={rgb,255: red,160; green,160; blue,160}, ->, line width=0.5mm},
invalid/.style = {->, line width=0.5mm, draw=red},
deadend/.style = {fill=red, draw=red, shape=circle, minimum size=2.5mm, inner sep=0mm}
]
	\begin{pgfonlayer}{nodelayer}
		\draw (0, 6) node [vertex, label=left:{$s$}] (0) {};
	
        \draw (0, 4) node [vertex, label={[label distance=-2mm]above right:{(1,0)}}] (1)  {};
		\draw (2, 4) node [vertex, label=right:{(1,1)}] (2)  {};
		\draw (4, 4) node [vertex, label=right:{(1,2)}] (3)  {};
  
		\draw (2, 2) node [vertex, label=right:{(2,1)}] (7)  {};
		\draw (4, 2) node [vertex, label=right:{(2,2)}] (8)  {};
  
		\draw (8, 0) node [vertex, label={[label distance=-0.5mm]right:{$d$}}] (11) {};
	\end{pgfonlayer}
	\begin{pgfonlayer}{edgelayer}
		\draw [edge] (0) -- (2);
		\draw [edge] (0) -- (3);
  
		\draw [edge] (1) -- (7);
		\draw [edge] (1) -- (8);

        \draw [edge] (2) -- (7);
		\draw [edge] (2) -- (8);

        \draw [edge] (3) -- (8);



        \draw [edge] (7) -- (11);
        \draw [edge] (8) -- (11);
	\end{pgfonlayer}
\end{tikzpicture}
}
\end{subfigure}

\caption{Pruning opponent decision graph $G^t_{3,4}$ (\cref{fig:graph_example}) given $\pi^t = \langle 1, 3, 2 \rangle$, $\mathcal{L}_p^t = \langle 0, 1, 0 \rangle$, and $\delta_p = 0$. Opponent allocation bounds are $\underline{\phi}^t = \langle 1, 0, 2 \rangle$ and $\overline{\phi}^t = \langle 4, 2, 3 \rangle$. \subref{subfig:bound pruning} Red and blue edges exceed feasible allocation lower and upper bounds, respectively. \subref{subfig:dead end pruning} Red vertices are dead-ends for $i=1$, while blue vertices are dead-ends for $i=2$ after pruning for $i=1$. \subref{subfig:final decision graph} The final pruned graph of feasible opponent decisions. Thus there are only 3 feasible decisions: $\langle 1, 0, 3 \rangle$, $\langle 1, 1, 2 \rangle$, and $\langle 2, 0, 2 \rangle$.}
\label{fig:pruning example}
\end{figure*}

We can use the decision graph $G_{K,N_p}$ to identify all feasible decisions $\pi\in\Pi_p$ by applying depth-first search to enumerate all paths starting at vertex $s$ and ending at vertex $d$. The weights of the edges observed, in order of traversal, represent the decision associated with each path explored. Note that the complexity of this computation will be proportional to the number of paths through the graph, which is $\mathcal{O}(2^{\min(K-1, N_p)})$ \cite{vuCombinatorialBanditsSequential2019}.

\subsection{Pruning for Feasible Opponent Decisions}

Under bandit or semi-bandit feedback, the player does not receive sufficient information to identify their opponent's true decision. As such, we propose that \CB decision graphs can be used to efficiently compute the set of \textit{feasible} decisions that an opponent may have made based on the information received by the player in a given round. While the problem is still non-polynomial in nature, we may significantly reduce its running time by using this approach.

Given $K$ battlefields, player $p$, round $t$, player decision $\pi^t$, and opponent resources $N_{p'}$, we can compute a decision graph $G^t_{K,N_{p'}}$ representing all decisions by the opponent that would produce the same feedback received by player $p$ in round $t$ (e.g., $L_p^t$ or $\mathcal{L}_p^t$). That is, every path from $s$ to $d$ through this feasible decision graph represents a decision $\phi$ which, when played against $\pi^t$, results in the same feedback observed by the player. This produces the set $\Phi_p^t$ defined in \cref{sec:estimating metrics}. For example, under bandit feedback we may identify that $\Phi_p^t = \{\phi\in\Phi_p\ |\ L_p(\pi^t,\phi) = L_p^t \}$.

The player can compute bounds on the feasible opponent decisions that would produce the observed information using any available information. By computing bounds on the feasible allocations for a given battlefield, edges can be removed from the opponent's feasible decision graph that represent impossible allocations. A visual example of the pruning process for a full round is shown in \cref{fig:pruning example}. Focusing on a specific battlefield, we may follow the example displayed in \cref{subfig:bound pruning}. In this example, a player that loses draws (i.e., $\delta_p=0$) allocated 3 resources to battlefield 2 and won. Thus the opponent must have allocated 2 or fewer resources to that battlefield. Thus, any edges entering a vertex representing battlefield 2 with a weight greater than 2 are invalid and can be pruned (indicated in blue in the middle layer of \cref{subfig:bound pruning}). In practice, the received information and how it can be used depends on the context. To address this, we propose a general technique that only requires bounds on the feasible opponent allocations, regardless of the feedback type. Therefore, the user only needs to implement a way compute the bounds on the opponent's feasible allocations.

Recall that, for any edge $v_{i-1,n'}\to v_{i, n} \in\mathcal{E}$, the opponent allocation $\phi_i$ to battlefield $i$ is $\phi_i=n-n'$. We denote the lower and upper bounds on the feasible values of $\phi_i$ by $\underline{\phi}_i$ and $\overline{\phi}_i$, respectively. That is, the feasible allocations by a player to battlefield $i$ are bounded such that $\phi_i\in[\underline{\phi}_i \twodots \overline{\phi}_i]$. Thus any edge $v_{i-1,n'}\to v_{i, n}$ is invalid if $\phi_i = n-n'<\underline{\phi}_i$ or $\phi_i = n-n'>\overline{\phi}_i$. This is shown in \cref{subfig:bound pruning}. 

Pruning edges using these bounds may create dead-ends (i.e., vertices with outdegree 0). Trivially, no paths to vertex $d$ exist which pass through a dead-end vertex. Therefore, we can prune all dead-end vertices (except for $d$) and any edges directed at them from the feasible decision graph. This may create new dead-end vertices. Therefore, pruning can be performed iteratively from $i=K-1$ to $i=0$, eliminating all dead-end vertices. This process is displayed in \cref{subfig:dead end pruning}. Afterwards, $d$ can be reached from any remaining vertex in the graph, as seen in \cref{subfig:final decision graph}. Thus all attempted paths starting at $s$ represent a valid decision. \cref{alg:dead end pruning} presents an efficient procedure for performing dead-end pruning.\footnoteref{foot:supplementary material} Although not necessary, we can prune vertices with indegree 0 (excluding $s$) to make every vertex reachable from $s$. However, this does not have any meaningful impact beyond potentially reducing the necessary memory space required to represent and store the graph.

\section{Bounding Opponent \CB Allocations}\label{sec:estimating opponent allocations}

To enhance the precision and computational efficiency of our payoff estimates for the \CB game, we consider the example of semi-bandit feedback to establish bounds on an opponent's feasible decisions. The computation of these limits is conducted on a per-round basis and is not dependent on previous rounds. For this reason, in this section we omit the round index $t$ from the notation introduced in \cref{sec:problem formulation}. Proofs for all proposed theorems and lemmas are presented in \cref{supp:proofs}.\footnoteref{foot:supplementary material} 

For a given round, the player receives a vector of scalar payoffs for each battlefield (semi-bandit feedback). The player aims to estimate the feasible bounds $\underline{\phi}_i$ and $\overline{\phi}_i$ on their opponent's true decision $\phi$. We denote any battlefield $i\in I$ that player $p$ loses to be in the set $\Lambda$, and any battlefield they win the be in the set $\Omega$. As described in \cref{sec:modeling with graphs}, the true opponent allocation $\phi_i$ must be bounded such that $\phi_i\in[\underline{\phi}_i \twodots \overline{\phi}_i]$.

Note that there are several valid values for $\underline{\phi}_i$ and $\overline{\phi}_i$. Therefore, our goal is to minimize the size of the enclosed range by identifying the tightest feasible bounds. This is achieved by maximizing $\underline{\phi}_i$ and minimizing $\overline{\phi}_i$. \cref{tab:bound formulas} summarizes the tightest generally applicable bounds derived in this section and their associated conditions.

\begin{table}[htb!p]
\centering
\begin{tabular}{Sc  Sc Sc}
\toprule
Bound & Value     & Condition \\ \midrule
\multirow{2}{*}{$\overline{\phi}_i$} & $\pi_i + \delta_p - 1$ & $i\in\Omega$ \\ \cdashline{2-3}
    & $N_{p'} - \sum_{\lambda\in\Lambda\setminus\{i\}}[\pi_\lambda+\delta_p]$ & $i\in\Lambda$ \\ \midrule
\multirow{2}{*}{$\underline{\phi}_i$} & $0$ & $i\in\Omega$ \\ \cdashline{2-3}
    & $\pi_i + \delta_p$ & $i\in\Lambda$ \\ 
\bottomrule
\end{tabular}
\caption{Bounds on feasible opponent allocation $\phi_i$ given $i\in I$ implemented in experimentation.}
\label{tab:bound formulas}
\end{table}

Using our earlier definitions, the following propositions follow trivially:
\begin{restatable}{prop}{propexpandedn}\label{prop:expanded N}
    $N_p = \sum_{\lambda\in\Lambda}\pi_\lambda + \sum_{\omega\in\Omega}\pi_\omega$ for any player~$p$.
\end{restatable}
\begin{restatable}{prop}{proplosebaselowerbound}\label{prop:lose base lower bound}
    $\underline{\phi}_\lambda = \pi_\lambda+\delta_p$ is a valid lower bound for any battlefield $\lambda\in\Lambda$.
\end{restatable}
\begin{restatable}{prop}{propwinbaselowerbound}\label{prop:win base lower bound}
    $\underline{\phi}_\omega = 0$ is a valid lower bound for any battlefield $\omega\in\Omega$.
\end{restatable}
\begin{restatable}{prop}{proplosebaseupperbound}\label{prop:lose base upper bound}
    $\overline{\phi}_\lambda = N_{p'}$ is a valid upper bound for any battlefield $\lambda\in\Lambda$.
\end{restatable}
\begin{restatable}{prop}{propwinbaseupperbound}\label{prop:win base upper bound}
    $\overline{\phi}_\omega = \pi_\omega+\delta_p-1$ is a valid upper bound for any battlefield $\omega\in\Omega$.
\end{restatable}

These lead to the following \lcnamecref{lem:general upper bound} regarding the upper bound on any opponent allocation $\phi_i$:
\begin{restatable}{lemma}{lemgeneralupperbound}\label{lem:general upper bound}
    $\overline{\phi}_i = N_{p'} - \sum_{\lambda\in\Lambda\setminus\{i\}}[\pi_\lambda+\delta_p]$ is a valid upper bound for any battlefield $i\in I$.
\end{restatable}

We want to determine when this bound is tighter than our existing bounds from \cref{prop:lose base upper bound,prop:win base upper bound}. This being the case, we want to identify on what, if any, conditions it is tighter. To do so, we propose the following \lcnamecrefs{lem:alt lose upper bound}, which are used later in \cref{thm:tight upper bound}:

\begin{restatable}{lemma}{lemaltloseupperbound}\label{lem:alt lose upper bound}
    Given $i\in\Lambda$, then $\overline{\phi}_i = N_{p'} - \sum_{\lambda\in\Lambda\setminus\{i\}}[\pi_\lambda+\delta_p]$ (\cref{lem:general upper bound}) is always an equivalent or tighter valid upper bound compared to $\overline{\phi}_i = N_{p'}$ (\cref{prop:lose base upper bound}).
\end{restatable}
\begin{restatable}{lemma}{lemaltwinupperbound}\label{lem:alt win upper bound}
    Given $i\in\Omega$, then $\overline{\phi}_i = N_{p'} - \sum_{\lambda\in\Lambda}[\pi_\lambda+\delta_p]$ (\cref{lem:general upper bound}) is a tighter valid upper bound compared to $\overline{\phi}_i = \pi_i+\delta_p-1$ (\cref{prop:win base upper bound}) if and only if $N_{p'} + 1  < \sum_{\lambda\in\Lambda\cup\{i\}}[\pi_\lambda+\delta_p]$. 
\end{restatable}

These \lcnamecrefs{lem:alt lose upper bound} can be combined to yield a tight upper bound  $\overline{\phi}_i$:

\begin{restatable}{theorem}{thmtightupperbound}\label{thm:tight upper bound}
    $\overline{\phi}_i = N_{p'} - \sum_{\lambda\in\Lambda\setminus\{i\}}[\pi_\lambda+\delta_p]$ is a tighter upper bound compared to the bounds in \cref{prop:lose base upper bound,prop:win base upper bound} given a battlefield $i\in I$ if and only if $i\in\Lambda$ or $N_{p'} + 1  < \sum_{\lambda\in\Lambda\cup\{i\}}[\pi_\lambda+\delta_p]$.
\end{restatable}

Using a similar approach to \cref{lem:general upper bound,thm:tight upper bound}, we identify the following \lcnamecref{lem:general lower bound} regarding the lower bound on any opponent allocation $\phi_i$:

\begin{restatable}{lemma}{thmgenerallowerbound}\label{lem:general lower bound}
    $\underline{\phi}_i = \pi_i\cdot\mathds{1}_\Lambda(i) + (|\Lambda|-1-\mathds{1}_\Lambda(i))\times \left(\sum_{\lambda\in\Lambda}[\pi_\lambda+\delta_p]- N_{p'}\right) - \sum_{\omega\in\Omega\setminus\{i\}}[\pi_\omega+\delta_p-1]$ is a valid lower bound given any battlefield $i\in I$.
\end{restatable}
$\mathds{1}_S(i)$ denotes the indicator function of whether $i$ belongs to set $S$. To identify when this bound is tighter than our existing bounds from \cref{prop:lose base lower bound,prop:win base lower bound}, we propose the following \lcnamecrefs{lem:lose alt lower bound}, which are used later in \cref{thm:tight lower bound}:

\begin{restatable}{lemma}{lemlosealtlowerbound}\label{lem:lose alt lower bound}
    Given a battlefield $i\in\Lambda$, the bound $\underline{\phi}_i = N_{p'}  - N_p + \pi_i - \delta_p + (K-1)(1-\delta_p)$ is valid and a tighter lower bound compared to the bound in \cref{prop:lose base lower bound} if and only if $\Lambda=\{i\}$ and $N_p + 2\delta_p < N_{p'} + (K-1)(1-\delta_p)$.
\end{restatable}
\begin{restatable}{lemma}{lemwinaltlowerbound}\label{lem:win alt lower bound}
    Given a battlefield $i\in\Omega$, the bound $\underline{\phi}_i = N_{p'} - N_p + \pi_i + (K-1)(1-\delta_p)$ is a valid and tighter lower bound compared to the bound in \cref{prop:win base lower bound} if and only if $\Lambda=\varnothing$ and $N_p < N_{p'} + \pi_i + (K-1)(1-\delta_p)$.
\end{restatable}

These \lcnamecrefs{lem:lose alt lower bound} can be combined to identify a tight lower bound $\underline{\phi}_i$:

\begin{restatable}{theorem}{thmtightlowerbound}\label{thm:tight lower bound}
    $\underline{\phi}_i = N_{p'}  - N_p + \pi_i - \delta_p\cdot\mathds{1}_\Lambda(i) + (K-1)(1-\delta_p)$ is a tighter lower bound compared to the bounds in \cref{prop:lose base lower bound,prop:win base lower bound} given a battlefield $i\in I$ if and only if $\Lambda \in \{\varnothing,\{i\}\}$ and $N_p + 2\delta_p\cdot\mathds{1}_\Lambda(i)< N_{p'} + \pi_i(1-\mathds{1}_\Lambda(i)) + (K-1)(1-\delta_p)$.
\end{restatable}

We note that the special condition identified in \cref{lem:alt win upper bound} and used in \cref{thm:tight upper bound}, and the bound identified in \cref{thm:tight lower bound} were not implemented in our experimentation (\cref{sec:empirical analysis}) and are omitted from \cref{tab:bound formulas}. This was done for the sake of simplicity and computational efficiency.

\section{Empirical Analysis of Estimate Quality}\label{sec:empirical analysis}


We conduct a proof-of-concept evaluation of our estimation metrics using simulated \textit{semi-bandit} \CB games. In our estimate calculations we use the graph-pruning approach proposed in \cref{sec:modeling with graphs} to narrow down the number of feasible opponent decisions. Leveraging the bounds identified in \cref{sec:estimating opponent allocations}, we greatly improve the efficiency of computing our estimates and, in theory, their accuracy by ignoring infeasible decisions.

To perform our simulations, we consider a uniformly random decision-maker and three online resource allocation algorithms: \ttsc{MARA} \cite{daganBetterResourceAllocation2018}, \ttsc{CUCB-DRA} \cite{zuoCombinatorialMultiarmedBandits2021}, and \ttsc{Edge} \cite{vuCombinatorialBanditsSequential2019}. These algorithms are used purely for the purpose of generating allocations with varying behaviors. For the purposes of this paper, \emph{we exclusively evaluate the accuracy of our estimators and do not consider the comparative performance of these allocation algorithms}.

In our experimentation, we denote players A and B such that $\delta_A=0$ and $\delta_B=1$. For each game we simulate $T=1000$ sequential rounds. The set of simulated game configurations is described in \cref{tab:experiment param combos}. We ensure that $N_A \geq N_B$ to prevent player B from having a significant advantage due to both winning draws and having an excess of resources. Every algorithm competes against every other algorithm for each game configuration, resulting in 16 matchups per configuration, making a total of 96 simulated games. Evaluating our estimates for both players gives us a total of 192 data points for each metric.

\begin{table}[t]
\centering
\begin{tabular}{ccccc}
\toprule
$K$ & $N_A$ & $N_B$ & $|\Pi_A|$ & $|\Pi_B|$ \\ \cmidrule(lr){0-2} \cmidrule(lr){4-5}
$3$ & $10$  & $10$  & $66$      & $66$      \\
$3$ & $15$  & $10$  & $136$     & $66$      \\
$3$ & $15$  & $15$  & $136$     & $136$     \\ 
$5$ & $15$  & $15$  & $3876$    & $3876$    \\
$5$ & $20$  & $15$  & $10626$   & $3876$    \\
$5$ & $20$  & $20$  & $10626$   & $10626$   \\ \bottomrule
\end{tabular}
\caption{Simulated \CB game configurations.}
\label{tab:experiment param combos}
\end{table}

We use small values of $K$ due to the complexity of computing our metrics and their estimates increasing exponentially with respect to the number of battlefields (as described in \cref{sec:modeling with graphs}). No repetitions were performed as the large number of rounds ($T=1000$) effectively serves the same purpose in this context.

Our algorithm implementation details are as follows: For \ttsc{MARA} we set the required input parameter $c=2.5$ to match the settings used by \citet{daganBetterResourceAllocation2018}. However, to convert from continuous allocations to discrete allocations, we used the procedure described in \cref{alg:discretization}.\footnoteref{foot:supplementary material} For \ttsc{CUCB-DRA}, we implement a naive oracle which selects a decision by directly computing the mean payoff of the algorithm's current estimates for a large sample of possible decisions and selecting the decision with the maximum believed mean payoff. For \ttsc{Edge} we used set the parameter $\gamma=0.25$ (i.e., a $25\%$ chance of exploring each round) and the exploration distribution $\mu$ to be uniform, i.e., $\mu=\mathcal{U}(\Pi_p)$.

To analyze the quality of our metrics, we analyze the error of our estimated metrics compared to their true counterparts (e.g., Observable Max Payoff versus Max Payoff). To evaluate the magnitude of our error, we utilize normalized root mean square error (NRMSE). That is, the RMSE divided by the of the mean true value:
\begin{equation}
    \text{NRMSE} \coloneqq \frac{1}{\text{mean}(y)} \sqrt{\frac{\sum_t (y'_t-y_t)^2}{N}}.
\end{equation}
To evaluate the deviation of our errors, we utilize relative residual standard deviation (RRSD). That is, the standard deviation of residuals divided by the mean of the true value:
\begin{equation}
    \text{RRSD} \coloneqq \frac{1}{\text{mean}(y)} \sqrt{\frac{\sum_t((y'_t-y_t)-\text{mean}(y'-y))^2}{N}},
\end{equation}
where $y_t$ is the true metric (i.e., Max Payoff or Expected Payoff) for round $t$ and $y'_t$ is the estimated metric (i.e., Observable Max Payoff, Supremum Payoff, or Expected Payoff).

We normalize by the mean true metric to compensate for artificial increases in the magnitude of error due to a larger number of battlefields. This is because the number of possible decisions grows exponentially with the number of battlefields, leading to a larger range of possible values and payoffs, but also greater magnitude of error and deviation in our estimates. The exact results of these simulations are provided in \crefrange{tab:first result table}{tab:last result table}.\footnoteref{foot:supplementary material} 

\subsection{Observable Max Payoff}



Across all matchups and game configurations, RMSE and RRSD of Observable Max Payoff were effectively zero. The maximum RMSE and RRSD observed throughout all simulations were 0.01 for both measures. Notably, these values were only observed for games with $K=5$ battlefields (\cref{subtab:non-zero max payoff 1,subtab:non-zero max payoff 2}). This corroborates our hypothesis that error should scale with the number of battlefields. These results demonstrate that Observable Max Payoff is highly effective at estimating the true value of Max Payoff. 

\subsection{Supremum Payoff}

For all but two game configurations, every matchup produced RMSE and RRSD of effectively zero for Supremum Payoff. However, player A observed notable errors in the game with configuration $K=5$, $N_A=15$, and $N_B=15$ (\cref{subtab:high supremum payoff A}), while player B observed similar behavior with configuration $K=5$, $N_A=20$, and $N_B=15$ (\cref{subtab:high supremum payoff B}). This behavior can be explained due by the inherently pessimistic nature of Supremum Payoff discussed in \cref{sec:estimating metrics}. Specifically, it becomes increasingly likely that there exists opponent decisions that would induce relatively small Max Payoff values as the number of feasible opponent decisions increases (as it does exponentially with respect to $K$, as discussed in \cref{sec:modeling with graphs}). Even so, the interpretation of Supremum Payoff as the worst-case Max Payoff and the fact that the observed errors are relatively small continues to highlight the usefulness of Supremum Payoff.

\subsection{Observable Expected Payoff}

Across nearly all matchups and game configurations, the observed NRMSE and RRSD of Observable Expected Payoff compared to true Expected Payoff are very low. NRMSE was below 0.20 in 186 out of the 192 data points, and 163 were below 0.15. Notably, all of the data points with NRMSE greater than 0.20 occurred in matchups with the \ttsc{CUCB-DRA} algorithm. The maximum observed NRMSE was 0.259. The fact that \ttsc{CUCB-DRA} is highly non-uniform in the decision space suggests that the larger errors are mostly due to the UDA. Even so, the associated error is relatively small, suggesting the UDA's impact to be minor. Additionally, in 183 out of 192 data points RRSD was below 0.15, and 156 were below 0.10. The maximum observed RRSD was 0.170. This indicates that Observable Expected Payoff is a good proxy for the true value of Expected Payoff.

\section{Conclusion}\label{sec:conclusion}

In this paper, we developed a general definition for payoff in the context of mutually adversarial games, as well as a number of useful performance evaluation metrics and means for approximating them which utilize general payoff. Under the context of the \CB game, we proposed an efficient approach for identifying feasible opponent decisions from observed feedback to improve the accuracy of our payoff metric estimates. To that end, the existence and usability of a number of bounds on opponent actions based on semi-bandit feedback for the \CB game are proven for use with the decision graph pruning process.



\bibliography{references}

\begin{thebibliography}{11}
\providecommand{\natexlab}[1]{#1}

\bibitem[{Audibert, Bubeck, and
  Lugosi(2014)}]{audibertRegretOnlineCombinatorial2014}
Audibert, J.-Y.; Bubeck, S.; and Lugosi, G. 2014.
\newblock Regret in {{Online Combinatorial Optimization}}.
\newblock \emph{Mathematics of Operations Research}, 39(1): 31--45.

\bibitem[{Auer et~al.(2012)Auer, {Cesa-Bianchi}, Freund, and
  Schapire}]{auerNonstochasticMultiarmedBandit2012}
Auer, P.; {Cesa-Bianchi}, N.; Freund, Y.; and Schapire, R.~E. 2012.
\newblock The {{Nonstochastic Multiarmed Bandit Problem}}.
\newblock \emph{SIAM Journal on Computing}.

\bibitem[{{Cesa-Bianchi} and
  Lugosi(2012)}]{cesa-bianchiCombinatorialBandits2012}
{Cesa-Bianchi}, N.; and Lugosi, G. 2012.
\newblock Combinatorial Bandits.
\newblock \emph{Journal of Computer and System Sciences}, 78(5): 1404--1422.

\bibitem[{Dagan and Koby(2018)}]{daganBetterResourceAllocation2018}
Dagan, Y.; and Koby, C. 2018.
\newblock A {{Better Resource Allocation Algorithm}} with {{Semi-Bandit
  Feedback}}.
\newblock In \emph{Proceedings of {{Algorithmic Learning Theory}}}, 268--320.
  {PMLR}.

\bibitem[{Koc{\'a}k et~al.(2014)Koc{\'a}k, Neu, Valko, and
  Munos}]{kocakEfficientLearningImplicit2014}
Koc{\'a}k, T.; Neu, G.; Valko, M.; and Munos, R. 2014.
\newblock Efficient Learning by Implicit Exploration in Bandit Problems with
  Side Observations.
\newblock In \emph{Advances in {{Neural Information Processing Systems}}},
  volume~27. {Curran Associates, Inc.}

\bibitem[{Roberson(2006)}]{robersonColonelBlottoGame2006}
Roberson, B. 2006.
\newblock The {{Colonel Blotto}} Game.
\newblock \emph{Economic Theory}, 29(1): 1--24.

\bibitem[{Schwartz, Loiseau, and
  Sastry(2014)}]{schwartzHeterogeneousColonelBlotto2014}
Schwartz, G.; Loiseau, P.; and Sastry, S.~S. 2014.
\newblock The Heterogeneous {{Colonel Blotto}} Game.
\newblock In \emph{2014 7th {{International Conference}} on {{NETwork Games}},
  {{COntrol}} and {{OPtimization}} ({{NetGCoop}})}, 232--238.

\bibitem[{Vu(2020)}]{vuModelsSolutionsStrategic2020}
Vu, D.~Q. 2020.
\newblock \emph{Models and {{Solutions}} of {{Strategic Resource Allocation
  Problems}}: {{Approximate Equilibrium}} and {{Online Learning}} in {{Blotto
  Games}}}.
\newblock Ph.D. thesis, Sorbonne Universites, UPMC University of Paris 6.

\bibitem[{Vu, Loiseau, and Silva(2019)}]{vuCombinatorialBanditsSequential2019}
Vu, D.~Q.; Loiseau, P.; and Silva, A. 2019.
\newblock Combinatorial {{Bandits}} for {{Sequential Learning}} in {{Colonel
  Blotto Games}}.
\newblock In \emph{2019 {{IEEE}} 58th {{Conference}} on {{Decision}} and
  {{Control}} ({{CDC}})}, 867--872.

\bibitem[{Vu et~al.(2020)Vu, Loiseau, Silva, and
  {Tran-Thanh}}]{vuPathPlanningProblems2020}
Vu, D.~Q.; Loiseau, P.; Silva, A.; and {Tran-Thanh}, L. 2020.
\newblock Path {{Planning Problems}} with {{Side
  Observations}}\textemdash{{When Colonels Play Hide-and-Seek}}.
\newblock \emph{Proceedings of the AAAI Conference on Artificial Intelligence},
  34(02): 2252--2259.

\bibitem[{Zuo and {Joe-Wong}(2021)}]{zuoCombinatorialMultiarmedBandits2021}
Zuo, J.; and {Joe-Wong}, C. 2021.
\newblock Combinatorial {{Multi-armed Bandits}} for {{Resource Allocation}}.
\newblock In \emph{2021 55th {{Annual Conference}} on {{Information Sciences}}
  and {{Systems}} ({{CISS}})}, 1--4.

\end{thebibliography}

\renewcommand{\thesection}{S\arabic{section}}
\renewcommand{\thetable}{S\arabic{table}}
\renewcommand{\thefigure}{S\arabic{figure}}

\onecolumn

\begin{center}
\textbf{\large Supplemental Materials: General Performance Evaluation for Competitive Resource Allocation Games via Unseen Payoff Estimation}
\end{center}

\section{\CB Max Payoff Algorithm}\label{supp:max payoff alg}


Calculating the Max Payoff for a specific decision $\phi$ is trivial for the \CB game, as described in \cref{alg:best_payoff}. The optimal strategy is the greedy algorithm: Starting from the battlefield in which the opponent allocated the least resources, allocate the minimum number of resources required to win; repeat this process until the player does not have enough resources to overcome the opponent. Note that if the allocations are already sorted, this process can be done in linear time $\mathcal{O}(K)$. However, if sorting is required, the process becomes limited by the time complexity of the sorting algorithm (typically $\mathcal{O}(K\log K)$).

\begin{algorithm}[htb!p]
\centering
\caption{Max Payoff computation}
\label{alg:best_payoff}
\begin{algorithmic}[1]

\REQUIRE Opponent decision $\phi$, player resources $N$, player $p$
\ENSURE The maximum achievable payoff

\STATE $L \gets 0$
\FOR{$\phi_i\in\phi$ in ascending order of $\phi_i$}
    \IF{$N > \phi_i - \delta_p$}
        \STATE $L \gets L + 1$
        \STATE $N \gets N - \phi_i - \delta_p + 1$ 
    \ELSE
        \RETURN $L$
    \ENDIF
\ENDFOR
\RETURN $L$

\end{algorithmic}
\end{algorithm}

\section{Dead-End Pruning Algorithm}

For $K$ battlefields and $N$ resources available to the player associated with the graph, \cref{alg:dead end pruning} can be applied with complexity $\mathcal{O}(K\cdot N)$ provided that removing all edges entering a particular vertex from the set of edges can be done in constant time. If not, then the algorithm has worst-case complexity $\mathcal{O}(E)$, which is equivalent to $\mathcal{O}(K\cdot N^2)$.

\begin{algorithm}[htb!p]
\caption{Decision graph dead-end pruning}
\label{alg:dead end pruning}
\begin{algorithmic}[1]

\REQUIRE Decision graph $G$, battlefields $K$, resources $N$ 
\ENSURE Decision graph with dead-ends removed

\STATE $\mathcal{V} \gets \ttsc{Vertices}(G)$
\STATE $\mathcal{E} \gets \ttsc{Edges}(G)$
\FOR{$i\gets K-1\  \TO\ 0$}
    \FOR{$n \gets N\ \TO\ 0$}
        \IF{$\deg^+(v_{i,n}) = 0$}
            \STATE $\mathcal{V} \gets \mathcal{V} \setminus v_{i,n}$ 
            \STATE $\mathcal{E} \gets \mathcal{E} \setminus \ttsc{InEdges}(v_{i,n})$
        \ENDIF
    \ENDFOR
\ENDFOR
\RETURN $\ttsc{BuildGraph}(\mathcal{V}, \mathcal{E})$

\end{algorithmic}
\end{algorithm}

\section{\ttsc{MARA} Discretization}

First, \cref{alg:discretization} normalizes and scales the continuous allocation vector produced by \ttsc{MARA} with respect to the available resources $N$ (the $\ttsc{Normalize}$ function). Next, the floor function is applied to each element of the scaled allocation vector to get the integer portion of the allocations. The non-integer remainder is then normalized. If the allocated resources do not sum up to $N$, the normalized remainders are treated as probabilities from which to sample allocations without replacement and increment by 1 until all resources have been used. This treats the magnitude of the remainder as the proportional desire to allocate more resources to a particular battlefield. The probabilistic allocation of remaining resources is based on the assumption that, on average, the final discrete allocations will converge to the value of the continuous allocation times $N$.

{

\renewcommand{\algorithmicwhile}{\textbf{Function}}
\renewcommand{\algorithmicdo}{:}
\renewcommand{\algorithmicendwhile}[1]{\algorithmicend\ #1}

\begin{algorithm}[hbt!p]
\begin{algorithmic}[1]
\caption{Discretize continuous allocations}
\label{alg:discretization}

\REQUIRE Continuous allocations vector $X$, resources $N$
\ENSURE The discretized allocation vector

\WHILE{\ttsc{Discretize}($X,N$)}
    \STATE $X \gets N \cdot \ttsc{Normalize}(X)$
    \STATE $X' \gets \lfloor X \rfloor$
    \STATE $P \gets \ttsc{Normalize}(X\;\mathrm{mod}\;1)$
    \IF{$\sum X' > N$}
        \STATE Sample $N-\sum X'$ indices $I$ from $X'$ with probability distribution $P$ without replacement 
        \FOR{$i\in I$}
            \STATE $X'_i \gets X'_i + 1$
        \ENDFOR
    \ENDIF
    \RETURN $X'$
\ENDWHILE{\ttsc{Discretize}}

\WHILE{\ttsc{Normalize}($Y$)}
    \RETURN $Y / \sum Y$
\ENDWHILE{\ttsc{Normalize}}

\end{algorithmic}
\end{algorithm}

}

\section{Computation Hardware/Software}

Utilized compute nodes ran Ubuntu 20.04 LTS with AMD EPYC 7543 and Intel Xeon Gold 6248 CPUs, and 64 GB of random access memory.


\section{Proofs}\label{supp:proofs}

\propexpandedn*
\proplosebaselowerbound*
\propwinbaselowerbound*
\proplosebaseupperbound*
\propwinbaseupperbound*
\begin{prop}\label{prop:disjoint sets}
    $\Omega$ and $\Lambda$ are a partition of $I$ (i.e., $\Omega\cup\Lambda = I$ and $\Omega\cap\Lambda = \emptyset$).
\end{prop}

This leads to the following proofs regarding the upper bound on any opponent allocation $\phi_i$:

\lemgeneralupperbound*
\begin{proof}
    Consider that we can pull out $\phi_i$ from $\sum_{\lambda\in\Lambda}\pi_\lambda$ and $\sum_{\omega\in\Omega}\phi_\omega$ in \cref{prop:expanded N}, as $\phi_i$ must be present in either $\Omega$ \textit{or} $\Lambda$, but not both (\cref{prop:disjoint sets}). This gives us $N_{p'} = \phi_i + \sum_{\lambda\in\Lambda\setminus\{i\}}\phi_\lambda + \sum_{\omega\in\Omega\setminus\{i\}}\phi_\omega$. If we move $\phi_i$ to stand alone on the LHS, we get $\phi_i = N_{p'} - \sum_{\lambda\in\Lambda\setminus\{i\}}\phi_\lambda - \sum_{\omega\in\Omega\setminus\{i\}}\phi_\omega$. Notice that we may find a valid upper bound on $\phi_i$ by finding a minimal upper bound (i.e., supremum) on this equation, which we can do by using our existing lower bounds from \cref{prop:lose base lower bound,prop:win base lower bound}. That is,
    \begin{equation}
    \begin{split}
        \overline{\phi}_i &= \sup\left(N_{p'} - \sum_{\lambda\in\Lambda\setminus\{i\}}\phi_\lambda - \sum_{\omega\in\Omega\setminus\{i\}}\phi_\omega\right) \\
            &= N_{p'} - \sum_{\lambda\in\Lambda\setminus\{i\}}\underline{\phi}_\lambda - \sum_{\omega\in\Omega\setminus\{i\}}\underline{\phi}_\omega \\
            &= N_{p'} - \sum_{\lambda\in\Lambda\setminus\{i\}}[\pi_\lambda+\delta_p] - \sum_{\omega\in\Omega\setminus\{i\}}0 \\
            &= N_{p'} - \sum_{\lambda\in\Lambda\setminus\{i\}}[\pi_\lambda+\delta_p].
    \end{split}
    \end{equation}
    Thus $\phi_i \leq \overline{\phi}_i = N_{p'} - \sum_{\lambda\in\Lambda\setminus\{i\}}[\pi_\lambda+\delta_p]$.

    We can also verify this via contradiction: Consider the inverse case where $\phi_i > N_{p'} - \sum_{\lambda\in\Lambda\setminus\{i\}}[\pi_\lambda+\delta_p]$. This can be adjusted as follows:
    \begin{equation}
    \begin{alignedat}{3}
            && \phi_i > N_{p'} - \sum_{\lambda\in\Lambda\setminus\{i\}}[\pi_\lambda+\delta_p] \\
        \iff && \phi_i + \sum_{\lambda\in\Lambda\setminus\{i\}}[\pi_\lambda+\delta_p] >  N_{p'}.
    \end{alignedat}
    \end{equation}
    Recall that $\pi_\lambda + \delta_p \leq \phi_\lambda$ for all $i\in\Lambda$. Thus,
    \begin{equation}
    \begin{alignedat}{4}
            && \phi_i + \sum_{\lambda\in\Lambda\setminus\{i\}}\phi_i &\geq \phi_i + \sum_{\lambda\in\Lambda\setminus\{i\}}[\pi_\lambda+\delta_p] >  N_{p'} \\
        \implies && \phi_i + \sum_{\lambda\in\Lambda\setminus\{i\}}\phi_i &>  N_{p'} \\
        \iff && \phi_i + \sum_{\lambda\in\Lambda\setminus\{i\}}\phi_i &>  \sum_{\lambda\in\Lambda}\phi_\lambda + \sum_{\omega\in\Omega}\phi_\omega && \text{(by \cref{prop:expanded N}})\\
        \iff && \phi_i + \sum_{\lambda\in\Lambda\setminus\{i\}}\phi_i &>  \phi_i + \sum_{\lambda\in\Lambda\setminus\{i\}}\phi_\lambda + \sum_{\omega\in\Omega\setminus\{i\}}\phi_\omega \\
        \iff && 0 &> \sum_{\omega\in\Omega\setminus\{i\}}\phi_\omega.
    \end{alignedat}
    \end{equation}
    However, $\phi_i\geq0$ for all $i\in I$ by definition, therefore $0 \leq \sum_{\omega\in\Omega\setminus\{i\}}\phi_\omega$. Thus $0 \not >  \sum_{\omega\in\Omega\setminus\{i\}}\phi_\omega$. This is a contradiction, proving $\overline{\phi}_i = N_{p'} - \sum_{\lambda\in\Lambda\setminus\{i\}}[\pi_\lambda+\delta_p]$ is a valid upper bound.
\end{proof}

Note that we want to identify when this bound is tighter than our existing bounds from \cref{prop:lose base upper bound,prop:win base lower bound}. This being the case, we attempt to identify on what, if any, conditions it is tighter. We start by focusing on the case where $i\in\Lambda$, which will later be used in \cref{thm:tight lower bound}:

\lemaltloseupperbound*
\begin{proof}
    We wish to show that $N_{p'} - \sum_{\lambda\in\Lambda\setminus\{i\}}[\pi_\lambda+\delta_p] \leq N_{p'}$. This simplifies to $0  \leq \sum_{\lambda\in\Lambda\setminus\{i\}}[\pi_\lambda+\delta_p]$. Recall that $\delta_p\in\{0,1\}$ and $\pi_i\geq 0$ by definition. Therefore we can always ensure that $0 \leq \sum_{\lambda\in\Lambda\setminus\{i\}}[\pi_\lambda+\delta_p]$.
\end{proof}

We now focus on the case where $i\in\Omega$, which will also later be used in \cref{thm:tight lower bound}:

\lemaltwinupperbound*
\begin{proof}
    Notice that $i\in\Omega \iff i\not\in\Lambda$ (\cref{prop:disjoint sets}), therefore $\overline{\phi}_i = N_{p'} - \sum_{\lambda\in\Lambda\setminus\{i\}}[\pi_\lambda+\delta_p] = N_{p'} - \sum_{\lambda\in\Lambda}[\pi_\lambda+\delta_p]$ (\cref{lem:general upper bound}). In order to validate that this is a tighter valid upper bound compared to $\overline{\phi}_i = \pi_i+\delta_p-1$ (\cref{prop:win base upper bound}), we need the former to be less than the latter. That is,
    \begin{equation}
    \begin{alignedat}{3}
            && N_{p'} - \sum_{\lambda\in\Lambda}[\pi_\lambda+\delta_p] &<  \pi_i+\delta_p-1\\
        \iff && N_{p'} + 1  &< \pi_i + \delta_p + \sum_{\lambda\in\Lambda}[\pi_\lambda+\delta_p] \\
         \\
        \iff && N_{p'} + 1  &< \sum_{\lambda\in\Lambda\cup\{i\}}[\pi_\lambda+\delta_p]. 
    \end{alignedat}
    \end{equation}
\end{proof}

Notably, we can easily reconcile the upper bounds from \cref{lem:alt lose upper bound,lem:alt win upper bound}:

\thmtightupperbound*
\begin{proof}
    It follows trivially from combining \cref{lem:alt lose upper bound,lem:alt win upper bound}.
\end{proof}

Notice that we may use a similar approach to \cref{lem:general upper bound,thm:tight lower bound} to find potential lower bounds as well:

\thmgenerallowerbound*
\begin{proof}
    Consider that we can pull out $\phi_i$ from $\sum_{\lambda\in\Lambda}\pi_\lambda$ and $\sum_{\omega\in\Omega}\phi_\omega$ in \cref{prop:expanded N}, as $\phi_i$ must be present in either $\Omega$ \textit{or} $\Lambda$, but not both (\cref{prop:disjoint sets}). This gives us $N_{p'} = \phi_i + \sum_{\lambda\in\Lambda\setminus\{i\}}\phi_\lambda + \sum_{\omega\in\Omega\setminus\{i\}}\phi_\omega$. If we move $\phi_i$ to stand alone on the LHS, we get $\phi_i = N_{p'} - \sum_{\lambda\in\Lambda\setminus\{i\}}\phi_\lambda - \sum_{\omega\in\Omega\setminus\{i\}}\phi_\omega$. Notice that we may find a valid lower bound on $\phi_i$ by finding a maximal lower bound (i.e., infimum) this equation, which we can do by using our existing upper bounds from \cref{lem:general upper bound,prop:win base upper bound}. That is,
    \begin{equation}\label{eq:general lower bound}
    \begin{split}
        \underline{\phi}_i &= \inf\left(N_{p'} - \sum_{\lambda\in\Lambda\setminus\{i\}}\phi_\lambda - \sum_{\omega\in\Omega\setminus\{i\}}\phi_\omega\right) \\
            &= N_{p'} - \sum_{\lambda\in\Lambda\setminus\{i\}}\overline{\phi}_\lambda - \sum_{\omega\in\Omega\setminus\{i\}}\overline{\phi}_\omega \\
            &= N_{p'} - \sum_{j\in\Lambda\setminus\{i\}}\left[N_{p'} - \sum_{\lambda\in\Lambda\setminus\{j\}}[\pi_\lambda+\delta_p]\right] - \sum_{\omega\in\Omega\setminus\{i\}}[\pi_\omega+\delta_p-1] \\
            &= N_{p'} - N_{p'}(|\Lambda|-\mathds{1}_\Lambda(i))  + \sum_{j\in\Lambda\setminus\{i\}}\sum_{\lambda\in\Lambda\setminus\{j\}}[\pi_\lambda+\delta_p] - \sum_{\omega\in\Omega\setminus\{i\}}[\pi_\omega+\delta_p-1] \\
            &= - N_{p'}(|\Lambda|-1-\mathds{1}_\Lambda(i))  + (|\Lambda|-1-\mathds{1}_\Lambda(i))\sum_{\lambda\in\Lambda}[\pi_\lambda+\delta_p] + \pi_i\cdot\mathds{1}_\Lambda(i) - \sum_{\omega\in\Omega\setminus\{i\}}[\pi_\omega+\delta_p-1] \\
            &= \pi_i\cdot\mathds{1}_\Lambda(i) + (|\Lambda|-1-\mathds{1}_\Lambda(i))\left(\sum_{\lambda\in\Lambda}[\pi_\lambda+\delta_p]- N_{p'}\right) - \sum_{\omega\in\Omega\setminus\{i\}}[\pi_\omega+\delta_p-1].
    \end{split}
    \end{equation}
    
    Thus $\underline{\phi}_i = \pi_i\cdot\mathds{1}_\Lambda(i) + (|\Lambda|-1-\mathds{1}_\Lambda(i))\left(\sum_{\lambda\in\Lambda}[\pi_\lambda+\delta_p]- N_{p'}\right) - \sum_{\omega\in\Omega\setminus\{i\}}[\pi_\omega+\delta_p-1]$ is a valid lower bound for any $i\in I$.
\end{proof}

The following \lcnamecref{lem:win geq 0} will be used in \cref{lem:lose alt lower bound,lem:win alt lower bound}:
\begin{lemma}\label{lem:win geq 0}
    $\pi_\omega+\delta_p-1 \geq 0$ for any battlefield $\omega\in\Omega$.
\end{lemma}
\begin{proof}[Proof by contradiction]
    Suppose $\pi_\omega+\delta_p-1 < 0$. This gives two cases:
    \begin{description}
        \item[Case 1:]
        Let $\delta_p=1$ (i.e., $p$ wins draws), thus $\pi_\omega < 0$. However, $\pi_i\geq 0$ for all $i\in I$ by definition. This is a contradiction
        \item[Case 2:]
        Let $\delta_p=0$ (i.e., $p$ loses draws), thus $\pi_\omega < 1$. Because $\pi_i\in\mathbb{N}_0$ for all $i\in I$ by definition, that means $\pi_\omega=0$. However, if $p$ won battlefield $w$, then $\pi_\omega+\delta>\phi_\omega$, or $\phi_\omega<\pi_\omega = 0$, thus $\phi_\omega < 0$. However, $\phi_i\in\mathbb{N}_0$ for all $i\in I$ by definition. This is a contradiction.        
    \end{description}
    Because both cases fail, $\pi_\omega+\delta_p-1 \not< 0$, therefore $\pi_\omega+\delta_p-1 \geq 0$ for any battlefield $\omega\in\Omega$.
\end{proof}

Given the bound identified in \cref{lem:general upper bound}, it is unintuitive whether this bound is ever tighter than our existing bounds from \cref{prop:lose base lower bound,prop:win base lower bound}. This being the case, we attempt to identify on what, if any, conditions it is tighter. We start by focusing on the case where $i\in\Lambda$, which will later be used in \cref{thm:tight upper bound}:

\lemlosealtlowerbound*
\begin{proof}
    Suppose $i\in\Lambda$, we want to compare the bound identified in \cref{lem:general lower bound} against $\underline{\phi}_i=\pi_i+\delta_p$ (\cref{prop:lose base lower bound}). Thus we wish to verify
    \begin{equation}
    \begin{alignedat}{3}
            && i\in\Lambda \wedge \pi_i+\delta_p &< \pi_i\cdot\mathds{1}_\Lambda(i) + (|\Lambda|-1-\mathds{1}_\Lambda(i))\left(\sum_{\lambda\in\Lambda}[\pi_\lambda+\delta_p]- N_{p'}\right) - \sum_{\omega\in\Omega\setminus\{i\}}[\pi_\omega+\delta_p-1] \\
        \implies && \pi_i+\delta_p &< N_{p'} + \pi_i + (|\Lambda|-2)\left(\sum_{\lambda\in\Lambda}[\pi_\lambda+\delta_p] - N_{p'}\right) - \sum_{\omega\in\Omega}[\pi_\omega+\delta_p-1] \\
        \iff && \delta_p &< N_{p'} + (|\Lambda|-2)\left(\sum_{\lambda\in\Lambda}[\pi_\lambda+\delta_p] - N_{p'}\right) - \sum_{\omega\in\Omega}[\pi_\omega+\delta_p-1].
    \end{alignedat}
    \end{equation}
    Notice that \cref{lem:win geq 0} implies that $- \sum_{\omega\in\Omega}[\pi_\omega+\delta_p-1] \leq 0$. Additionally, consider that we can expand $N_{p'}$ based on \cref{prop:expanded N}:
    \begin{equation}
    \begin{alignedat}{3}
            && \delta_p &< N_{p'} + (|\Lambda|-2)\left(\sum_{\lambda\in\Lambda}[\pi_\lambda+\delta_p] - N_{p'}\right) - \sum_{\omega\in\Omega}[\pi_\omega+\delta_p-1] \\
        \implies && \delta_p &< N_{p'} + (|\Lambda|-2)\left(\sum_{\lambda\in\Lambda}[\pi_\lambda+\delta_p]- N_{p'}\right) & \mathllap{(\text{by \cref{lem:win geq 0}})} \\
        \iff && \delta_p &< (|\Lambda|-2)\left(\sum_{\lambda\in\Lambda}[\pi_\lambda+\delta_p]- \sum_{\lambda\in\Lambda}\phi_\lambda - \sum_{\omega\in\Omega}\phi_\omega\right) \\ 
        \iff && \delta_p &< (|\Lambda|-2)\left(\sum_{\lambda\in\Lambda}[\pi_\lambda+\delta_p - \phi_\lambda] - \sum_{\omega\in\Omega}\phi_\omega\right). 
    \end{alignedat}
    \end{equation}
    Notice that $\pi_i+\delta_p\leq \phi_i$ for all $i\in\Lambda$ by definition, therefore $\sum_{\lambda\in\Lambda}[\pi_\lambda+\delta_p-\phi_\lambda]\leq 0$. Additionally, $\phi_i\geq 0$ for all $i\in I$ by definition, therefore $-\sum_{\omega\in\Omega}\phi_\omega \leq 0$. Thus $\sum_{\lambda\in\Lambda}[\pi_\lambda+\delta_p - \phi_\lambda] - \sum_{\omega\in\Omega}\phi_\lambda \leq 0$. Therefore because $\delta_p\in\{0,1\}$, the inequality can only hold if $|\Lambda|-2 < 0$, which, given $i\in\Lambda$ is only possible when $|\Lambda|=1$ (i.e., $\Lambda = \{i\}$).

    In the circumstance where $\Lambda = \{i\}$, we can simplify \cref{lem:general lower bound}:
    \begin{equation}
    \begin{alignedat}{3}
            && \Lambda = \{i\} \wedge \underline{\phi}_i &= \pi_i\cdot\mathds{1}_\Lambda(i) + (|\Lambda|-1-\mathds{1}_\Lambda(i))\left(\sum_{\lambda\in\Lambda}[\pi_\lambda+\delta_p]- N_{p'}\right) - \sum_{\omega\in\Omega\setminus\{i\}}[\pi_\omega+\delta_p-1] \\
        \implies && \underline{\phi}_i &= \pi_i - (\pi_i+\delta_p- N_{p'}) - \sum_{\omega\in\Omega}[\pi_\omega+\delta_p-1] \\
        \iff && & = N_{p'} - \delta_p + |\Omega|(1-\delta_p) - \sum_{\omega\in\Omega}\pi_\omega \\
        \iff && & = N_{p'} - \delta_p + |\Omega|(1-\delta_p) - (N_p - \pi_i) \\
        \iff && & = N_{p'}  - N_p + \pi_i - \delta_p + (K-1)(1-\delta_p) && \mathllap{(\Omega = I\setminus\{i\})}.
    \end{alignedat}
    \end{equation}
    We can also simplify the inequality for verifying this bound's usefulness compared against $\underline{\phi}_i=\pi_i+\delta_p$ (\cref{prop:lose base lower bound}):
    \begin{equation}
    \begin{alignedat}{3}
            && \pi_i+\delta_p &< N_{p'} - N_p + \pi_i - \delta_p + (K-1)(1-\delta_p) \\
        \iff && N_p + 2\delta_p &< N_{p'} + (K-1)(1-\delta_p) \\ 
    \end{alignedat}
    \end{equation}
    Therefore given $i\in\Lambda$, the bound $\underline{\phi}_i = N_{p'}  - N_p + \pi_i - \delta_p + (K-1)(1-\delta_p)$ is useful \textit{if and only if} $\Lambda = \{i\}$ and $N_p + 2\delta_p < N_{p'} + (K-1)(1-\delta_p)$.
\end{proof}

We now focus on the case where $i\in\Omega$, which will later be used in \cref{thm:tight upper bound}:

\lemwinaltlowerbound*
\begin{proof}
    Suppose $i\in\Omega$, we want to compare the bound identified in \cref{lem:general lower bound} against $\underline{\phi}_i=0$ (\cref{prop:win base lower bound}). Thus we wish to verify
    \begin{equation}
    \begin{alignedat}{3}
            && i\in\Omega \wedge 0 &< \pi_i\cdot\mathds{1}_\Lambda(i) + (|\Lambda|-1-\mathds{1}_\Lambda(i))\left(\sum_{\lambda\in\Lambda}[\pi_\lambda+\delta_p]- N_{p'}\right) - \sum_{\omega\in\Omega\setminus\{i\}}[\pi_\omega+\delta_p-1] \\
        \implies && 0 &< (|\Lambda|-1)\left(\sum_{\lambda\in\Lambda}[\pi_\lambda+\delta_p]- N_{p'}\right) - \sum_{\omega\in\Omega\setminus\{i\}}[\pi_\omega+\delta_p-1] \\
        \implies && 0 &< (|\Lambda|-1)\left(\sum_{\lambda\in\Lambda}[\pi_\lambda+\delta_p]- N_{p'}\right) & \mathllap{(\text{by \cref{lem:win geq 0}}).}
    \end{alignedat}
    \end{equation}
    
    Notice that we can expand $N_{p'}$ based on \cref{prop:expanded N}:
    \begin{equation}
    \begin{alignedat}{3}
            && 0 &< (|\Lambda|-1)\left(\sum_{\lambda\in\Lambda}[\pi_\lambda+\delta_p]- N_{p'}\right) \\
        \iff && 0 &< (|\Lambda|-1)\left(\sum_{\lambda\in\Lambda}[\pi_\lambda+\delta_p] - \sum_{\lambda\in\Lambda}\phi_\lambda - \sum_{\omega\in\Omega}\phi_i\right) \\
        \iff && 0 &< (|\Lambda|-1)\left(\sum_{\lambda\in\Lambda}[\pi_\lambda+\delta_p-\phi_\lambda] - \sum_{\omega\in\Omega}\phi_i\right).
    \end{alignedat}
    \end{equation}
    Recall that $\pi_i+\delta_p\leq \phi_i$ for all $i\in\Lambda$ by definition, therefore $\sum_{\lambda\in\Lambda}[\pi_\lambda+\delta_p-\phi_\lambda]\leq 0$. Additionally, $\phi_i\geq 0$ for all $i\in I$ by definition, therefore $-\sum_{\omega\in\Omega}\phi_i \leq 0$. Therefore the inequality can only hold if $|\Lambda|-1 < 0$, which is only possible when $|\Lambda|=0$, in which case $\Lambda=\varnothing$.
    
    In the circumstance where $i\in\Omega$ and $\Lambda=\varnothing$, we can simplify \cref{lem:general lower bound}:
    \begin{equation}\label{eq:alt lower bound}
    \begin{alignedat}{3}
            & \mathrlap{i\in\Omega \wedge \Lambda=\varnothing} \\
            && \quad \wedge\ \underline{\phi}_i &= \pi_i\cdot\mathds{1}_\Lambda(i) + (|\Lambda|-1-\mathds{1}_\Lambda(i))\left(\sum_{\lambda\in\Lambda}[\pi_\lambda+\delta_p]- N_{p'}\right) - \sum_{\omega\in\Omega\setminus\{i\}}[\pi_\omega+\delta_p-1] \\
        \implies && \underline{\phi}_i &= N_{p'} - \sum_{\omega\in\Omega\setminus\{i\}}[\pi_\omega+\delta_p-1] \\
        \iff && &= N_{p'} + (|\Omega|-1)(1-\delta_p) - \sum_{\omega\in\Omega\setminus\{i\}}\pi_\omega \\
        \iff && &= N_{p'} + (|\Omega|-1)(1-\delta_p) - \left(\sum_{\omega\in\Omega}\pi_\omega - \pi_i\right)\\
        \iff && &= N_{p'} + (|\Omega|-1)(1-\delta_p) - N_p + \pi_i \\
        \iff && &= N_{p'} - N_p + \pi_i + (K-1)(1-\delta_p) & \mathllap{(\Omega=I)}.
    \end{alignedat}
    \end{equation}
    We can also simplify the inequality for verifying this bound's usefulness compared against $\underline{\phi}_i=0$ (\cref{prop:win base lower bound}):
    \begin{equation}
    \begin{alignedat}{3}
            && 0 &< N_{p'} - N_p + \pi_i + (K-1)(1-\delta_p) \\
        \iff && N_p &< N_{p'} + \pi_i + (K-1)(1-\delta_p)
    \end{alignedat}
    \end{equation}
    Therefore given $i\in\Omega$, the bound $\underline{\phi}_i = N_{p'} - N_p + \pi_i + (K-1)(1-\delta_p)$ is useful \textit{if and only if} $\Lambda=\varnothing$ and $N_p < N_{p'} + \pi_i + (K-1)(1-\delta_p)$.
\end{proof}

Notably, we can easily reconcile the lower bounds from \cref{lem:lose alt lower bound,lem:win alt lower bound}:

\thmtightlowerbound*
\begin{proof}
    It follows trivially from combining \cref{lem:lose alt lower bound,lem:win alt lower bound}.
\end{proof}



\section{Empirical Results}


The tables in this section focus on the error with respect to each payoff estimation metric. The sub-tables show the following metrics computer over the course of the game: (a) The Normalized Root Mean-Squared Error (NRMSE) $\pm$ the Relative Residual Standard Deviation (RRSD) between Observable Expected Payoff and True Expected Payoff; (b) the NRMSE $\pm$ RRSD between Observable Max Payoff and True Max Payoff over the course of the game; (c) the NMRSE $\pm$ RRSD between Supremum Payoff and True Payoff. The values presented are observed by the player designated by the rows against the player designated by the columns.

\twocolumn

\begin{table}[htb!p]
\begin{subtable}[h]{\linewidth}
\centering
\caption{Observable Expected Payoff Normalized Error}
\resizebox{\linewidth}{!}{
\begin{tabular}{rcccc}
\toprule
 & \ttsc{MARA} & \ttsc{CUCB-DRA} & \ttsc{Edge} & Random \\
\midrule
\ttsc{MARA} & $.027\pm.019$ & $.114\pm.091$ & $.085\pm.085$ & $.088\pm.088$ \\
\ttsc{CUCB-DRA} & $.163\pm.078$ & $.096\pm.093$ & $.085\pm.085$ & $.086\pm.086$ \\
\ttsc{Edge} & $.127\pm.068$ & $.107\pm.094$ & $.081\pm.081$ & $.080\pm.080$ \\
Random & $.119\pm.076$ & $.107\pm.092$ & $.082\pm.082$ & $.082\pm.082$ \\
\bottomrule
\end{tabular}
}
\end{subtable}

\bigskip

\begin{subtable}[h]{\linewidth}
\centering
\caption{Observable Max Payoff Normalized Error}
\resizebox{\linewidth}{!}{
\begin{tabular}{rcccc}
\toprule
 & \ttsc{MARA} & \ttsc{CUCB-DRA} & \ttsc{Edge} & Random \\
\midrule
\ttsc{MARA} & $.000\pm.000$ & $.000\pm.000$ & $.000\pm.000$ & $.000\pm.000$ \\
\ttsc{CUCB-DRA} & $.000\pm.000$ & $.000\pm.000$ & $.000\pm.000$ & $.000\pm.000$ \\
\ttsc{Edge} & $.000\pm.000$ & $.000\pm.000$ & $.000\pm.000$ & $.000\pm.000$ \\
Random & $.000\pm.000$ & $.000\pm.000$ & $.000\pm.000$ & $.000\pm.000$ \\
\bottomrule
\end{tabular}
}
\end{subtable}

\bigskip

\begin{subtable}[h]{\linewidth}
\centering
\caption{Supremum Payoff Normalized Error}
\resizebox{\linewidth}{!}{
\begin{tabular}{rcccc}
\toprule
 & \ttsc{MARA} & \ttsc{CUCB-DRA} & \ttsc{Edge} & Random \\
\midrule
\ttsc{MARA} & $.000\pm.000$ & $.000\pm.000$ & $.000\pm.000$ & $.000\pm.000$ \\
\ttsc{CUCB-DRA} & $.000\pm.000$ & $.000\pm.000$ & $.000\pm.000$ & $.000\pm.000$ \\
\ttsc{Edge} & $.000\pm.000$ & $.000\pm.000$ & $.000\pm.000$ & $.000\pm.000$ \\
Random & $.000\pm.000$ & $.000\pm.000$ & $.000\pm.000$ & $.000\pm.000$ \\
\bottomrule
\end{tabular}
}
\end{subtable}
\caption{Empirical results focusing on player A (rows) versus player B (columns) for games with $T=1000$, $K=3$, $N_A=10$, and $N_B=10$.}
\label{tab:first result table}
\end{table}


\begin{table}[htb!p]
\begin{subtable}[h]{\linewidth}
\centering
\caption{Observable Expected Payoff Normalized Error}
\resizebox{\linewidth}{!}{
\begin{tabular}{rcccc}
\toprule
 & \ttsc{MARA} & \ttsc{CUCB-DRA} & \ttsc{Edge} & Random \\
\midrule
\ttsc{MARA} & $.016\pm.016$ & $.074\pm.069$ & $.058\pm.058$ & $.059\pm.059$ \\
\ttsc{CUCB-DRA} & $.133\pm.057$ & $.088\pm.079$ & $.063\pm.063$ & $.061\pm.061$ \\
\ttsc{Edge} & $.130\pm.056$ & $.089\pm.073$ & $.067\pm.067$ & $.064\pm.064$ \\
Random & $.123\pm.060$ & $.083\pm.073$ & $.063\pm.063$ & $.066\pm.066$ \\
\bottomrule
\end{tabular}
}
\end{subtable}

\bigskip

\begin{subtable}[h]{\linewidth}
\centering
\caption{Observable Max Payoff Normalized Error}
\resizebox{\linewidth}{!}{
\begin{tabular}{rcccc}
\toprule
 & \ttsc{MARA} & \ttsc{CUCB-DRA} & \ttsc{Edge} & Random \\
\midrule
\ttsc{MARA} & $.000\pm.000$ & $.000\pm.000$ & $.000\pm.000$ & $.000\pm.000$ \\
\ttsc{CUCB-DRA} & $.000\pm.000$ & $.000\pm.000$ & $.000\pm.000$ & $.000\pm.000$ \\
\ttsc{Edge} & $.000\pm.000$ & $.000\pm.000$ & $.000\pm.000$ & $.000\pm.000$ \\
Random & $.000\pm.000$ & $.000\pm.000$ & $.000\pm.000$ & $.000\pm.000$ \\
\bottomrule
\end{tabular}
}
\end{subtable}

\bigskip

\begin{subtable}[h]{\linewidth}
\centering
\caption{Supremum Payoff Normalized Error}
\resizebox{\linewidth}{!}{
\begin{tabular}{rcccc}
\toprule
 & \ttsc{MARA} & \ttsc{CUCB-DRA} & \ttsc{Edge} & Random \\
\midrule
\ttsc{MARA} & $.000\pm.000$ & $.000\pm.000$ & $.000\pm.000$ & $.000\pm.000$ \\
\ttsc{CUCB-DRA} & $.000\pm.000$ & $.000\pm.000$ & $.000\pm.000$ & $.000\pm.000$ \\
\ttsc{Edge} & $.000\pm.000$ & $.000\pm.000$ & $.000\pm.000$ & $.000\pm.000$ \\
Random & $.000\pm.000$ & $.000\pm.000$ & $.000\pm.000$ & $.000\pm.000$ \\
\bottomrule
\end{tabular}
}
\end{subtable}
\caption{Empirical results focusing on player B (rows) versus player A (columns) for games with $T=1000$, $K=3$, $N_A=10$, and $N_B=10$.}
\end{table}


\begin{table}[htb!p]
\begin{subtable}[h]{\linewidth}
\centering
\caption{Observable Expected Payoff Normalized Error}
\resizebox{\linewidth}{!}{
\begin{tabular}{rcccc}
\toprule
 & \ttsc{MARA} & \ttsc{CUCB-DRA} & \ttsc{Edge} & Random \\
\midrule
\ttsc{MARA} & $.014\pm.011$ & $.038\pm.033$ & $.033\pm.033$ & $.030\pm.030$ \\
\ttsc{CUCB-DRA} & $.076\pm.028$ & $.040\pm.038$ & $.032\pm.032$ & $.032\pm.032$ \\
\ttsc{Edge} & $.072\pm.029$ & $.040\pm.037$ & $.030\pm.030$ & $.031\pm.031$ \\
Random & $.072\pm.031$ & $.037\pm.037$ & $.030\pm.030$ & $.030\pm.030$ \\
\bottomrule
\end{tabular}
}
\end{subtable}

\bigskip

\begin{subtable}[h]{\linewidth}
\centering
\caption{Observable Max Payoff Normalized Error}
\resizebox{\linewidth}{!}{
\begin{tabular}{rcccc}
\toprule
 & \ttsc{MARA} & \ttsc{CUCB-DRA} & \ttsc{Edge} & Random \\
\midrule
\ttsc{MARA} & $.000\pm.000$ & $.000\pm.000$ & $.000\pm.000$ & $.000\pm.000$ \\
\ttsc{CUCB-DRA} & $.000\pm.000$ & $.000\pm.000$ & $.000\pm.000$ & $.000\pm.000$ \\
\ttsc{Edge} & $.000\pm.000$ & $.000\pm.000$ & $.000\pm.000$ & $.000\pm.000$ \\
Random & $.000\pm.000$ & $.000\pm.000$ & $.000\pm.000$ & $.000\pm.000$ \\
\bottomrule
\end{tabular}
}
\end{subtable}

\bigskip

\begin{subtable}[h]{\linewidth}
\centering
\caption{Supremum Payoff Normalized Error}
\resizebox{\linewidth}{!}{
\begin{tabular}{rcccc}
\toprule
 & \ttsc{MARA} & \ttsc{CUCB-DRA} & \ttsc{Edge} & Random \\
\midrule
\ttsc{MARA} & $.000\pm.000$ & $.000\pm.000$ & $.000\pm.000$ & $.000\pm.000$ \\
\ttsc{CUCB-DRA} & $.000\pm.000$ & $.000\pm.000$ & $.000\pm.000$ & $.000\pm.000$ \\
\ttsc{Edge} & $.000\pm.000$ & $.000\pm.000$ & $.000\pm.000$ & $.000\pm.000$ \\
Random & $.000\pm.000$ & $.000\pm.000$ & $.000\pm.000$ & $.000\pm.000$ \\
\bottomrule
\end{tabular}
}
\end{subtable}
\caption{Empirical results focusing on player A (rows) versus player B (columns) for games with $T=1000$, $K=3$, $N_A=15$, and $N_B=10$.}
\end{table}


\begin{table}[htb!p]
\begin{subtable}[h]{\linewidth}
\centering
\caption{Observable Expected Payoff Normalized Error}
\resizebox{\linewidth}{!}{
\begin{tabular}{rcccc}
\toprule
 & \ttsc{MARA} & \ttsc{CUCB-DRA} & \ttsc{Edge} & Random \\
\midrule
\ttsc{MARA} & $.139\pm.055$ & $.177\pm.170$ & $.156\pm.156$ & $.155\pm.155$ \\
\ttsc{CUCB-DRA} & $.112\pm.106$ & $.180\pm.163$ & $.158\pm.158$ & $.152\pm.152$ \\
\ttsc{Edge} & $.104\pm.103$ & $.187\pm.146$ & $.158\pm.158$ & $.156\pm.156$ \\
Random & $.101\pm.100$ & $.184\pm.148$ & $.149\pm.149$ & $.155\pm.155$ \\
\bottomrule
\end{tabular}
}
\end{subtable}

\bigskip

\begin{subtable}[h]{\linewidth}
\centering
\caption{Observable Max Payoff Normalized Error}
\resizebox{\linewidth}{!}{
\begin{tabular}{rcccc}
\toprule
 & \ttsc{MARA} & \ttsc{CUCB-DRA} & \ttsc{Edge} & Random \\
\midrule
\ttsc{MARA} & $.000\pm.000$ & $.000\pm.000$ & $.000\pm.000$ & $.000\pm.000$ \\
\ttsc{CUCB-DRA} & $.000\pm.000$ & $.000\pm.000$ & $.000\pm.000$ & $.000\pm.000$ \\
\ttsc{Edge} & $.000\pm.000$ & $.000\pm.000$ & $.000\pm.000$ & $.000\pm.000$ \\
Random & $.000\pm.000$ & $.000\pm.000$ & $.000\pm.000$ & $.000\pm.000$ \\
\bottomrule
\end{tabular}
}
\end{subtable}

\bigskip

\begin{subtable}[h]{\linewidth}
\centering
\caption{Supremum Payoff Normalized Error}
\resizebox{\linewidth}{!}{
\begin{tabular}{rcccc}
\toprule
 & \ttsc{MARA} & \ttsc{CUCB-DRA} & \ttsc{Edge} & Random \\
\midrule
\ttsc{MARA} & $.000\pm.000$ & $.000\pm.000$ & $.000\pm.000$ & $.000\pm.000$ \\
\ttsc{CUCB-DRA} & $.000\pm.000$ & $.000\pm.000$ & $.000\pm.000$ & $.000\pm.000$ \\
\ttsc{Edge} & $.000\pm.000$ & $.000\pm.000$ & $.000\pm.000$ & $.000\pm.000$ \\
Random & $.000\pm.000$ & $.000\pm.000$ & $.000\pm.000$ & $.000\pm.000$ \\
\bottomrule
\end{tabular}
}
\end{subtable}
\caption{Empirical results focusing on player B (rows) versus player A (columns) for games with $T=1000$, $K=3$, $N_A=15$, and $N_B=10$.}
\end{table}


\begin{table}[htb!p]
\begin{subtable}[h]{\linewidth}
\centering
\caption{Observable Expected Payoff Normalized Error}
\resizebox{\linewidth}{!}{
\begin{tabular}{rcccc}
\toprule
 & \ttsc{MARA} & \ttsc{CUCB-DRA} & \ttsc{Edge} & Random \\
\midrule
\ttsc{MARA} & $.029\pm.014$ & $.112\pm.103$ & $.083\pm.083$ & $.082\pm.082$ \\
\ttsc{CUCB-DRA} & $.175\pm.082$ & $.102\pm.097$ & $.081\pm.081$ & $.080\pm.080$ \\
\ttsc{Edge} & $.110\pm.061$ & $.113\pm.093$ & $.079\pm.079$ & $.077\pm.077$ \\
Random & $.121\pm.064$ & $.102\pm.088$ & $.076\pm.076$ & $.079\pm.079$ \\
\bottomrule
\end{tabular}
}
\end{subtable}

\bigskip

\begin{subtable}[h]{\linewidth}
\centering
\caption{Observable Max Payoff Normalized Error}
\resizebox{\linewidth}{!}{
\begin{tabular}{rcccc}
\toprule
 & \ttsc{MARA} & \ttsc{CUCB-DRA} & \ttsc{Edge} & Random \\
\midrule
\ttsc{MARA} & $.000\pm.000$ & $.000\pm.000$ & $.000\pm.000$ & $.000\pm.000$ \\
\ttsc{CUCB-DRA} & $.000\pm.000$ & $.000\pm.000$ & $.000\pm.000$ & $.000\pm.000$ \\
\ttsc{Edge} & $.000\pm.000$ & $.000\pm.000$ & $.000\pm.000$ & $.000\pm.000$ \\
Random & $.000\pm.000$ & $.000\pm.000$ & $.000\pm.000$ & $.000\pm.000$ \\
\bottomrule
\end{tabular}
}
\end{subtable}

\bigskip

\begin{subtable}[h]{\linewidth}
\centering
\caption{Supremum Payoff Normalized Error}
\resizebox{\linewidth}{!}{
\begin{tabular}{rcccc}
\toprule
 & \ttsc{MARA} & \ttsc{CUCB-DRA} & \ttsc{Edge} & Random \\
\midrule
\ttsc{MARA} & $.000\pm.000$ & $.000\pm.000$ & $.000\pm.000$ & $.000\pm.000$ \\
\ttsc{CUCB-DRA} & $.000\pm.000$ & $.000\pm.000$ & $.000\pm.000$ & $.000\pm.000$ \\
\ttsc{Edge} & $.000\pm.000$ & $.000\pm.000$ & $.000\pm.000$ & $.000\pm.000$ \\
Random & $.000\pm.000$ & $.000\pm.000$ & $.000\pm.000$ & $.000\pm.000$ \\
\bottomrule
\end{tabular}
}
\end{subtable}
\caption{Empirical results focusing on player A (rows) versus player B (columns) for games with $T=1000$, $K=3$, $N_A=15$, and $N_B=15$.}
\end{table}


\begin{table}[htb!p]
\begin{subtable}[h]{\linewidth}
\centering
\caption{Observable Expected Payoff Normalized Error}
\resizebox{\linewidth}{!}{
\begin{tabular}{rcccc}
\toprule
 & \ttsc{MARA} & \ttsc{CUCB-DRA} & \ttsc{Edge} & Random \\
\midrule
\ttsc{MARA} & $.015\pm.010$ & $.084\pm.080$ & $.064\pm.064$ & $.066\pm.066$ \\
\ttsc{CUCB-DRA} & $.135\pm.067$ & $.096\pm.087$ & $.062\pm.062$ & $.065\pm.065$ \\
\ttsc{Edge} & $.128\pm.060$ & $.093\pm.081$ & $.070\pm.070$ & $.066\pm.066$ \\
Random & $.119\pm.059$ & $.093\pm.082$ & $.066\pm.066$ & $.067\pm.067$ \\
\bottomrule
\end{tabular}
}
\end{subtable}

\bigskip

\begin{subtable}[h]{\linewidth}
\centering
\caption{Observable Max Payoff Normalized Error}
\resizebox{\linewidth}{!}{
\begin{tabular}{rcccc}
\toprule
 & \ttsc{MARA} & \ttsc{CUCB-DRA} & \ttsc{Edge} & Random \\
\midrule
\ttsc{MARA} & $.000\pm.000$ & $.000\pm.000$ & $.000\pm.000$ & $.000\pm.000$ \\
\ttsc{CUCB-DRA} & $.000\pm.000$ & $.000\pm.000$ & $.000\pm.000$ & $.000\pm.000$ \\
\ttsc{Edge} & $.000\pm.000$ & $.000\pm.000$ & $.000\pm.000$ & $.000\pm.000$ \\
Random & $.000\pm.000$ & $.000\pm.000$ & $.000\pm.000$ & $.000\pm.000$ \\
\bottomrule
\end{tabular}
}
\end{subtable}

\bigskip

\begin{subtable}[h]{\linewidth}
\centering
\caption{Supremum Payoff Normalized Error}
\resizebox{\linewidth}{!}{
\begin{tabular}{rcccc}
\toprule
 & \ttsc{MARA} & \ttsc{CUCB-DRA} & \ttsc{Edge} & Random \\
\midrule
\ttsc{MARA} & $.000\pm.000$ & $.000\pm.000$ & $.000\pm.000$ & $.000\pm.000$ \\
\ttsc{CUCB-DRA} & $.000\pm.000$ & $.000\pm.000$ & $.000\pm.000$ & $.000\pm.000$ \\
\ttsc{Edge} & $.000\pm.000$ & $.000\pm.000$ & $.000\pm.000$ & $.000\pm.000$ \\
Random & $.000\pm.000$ & $.000\pm.000$ & $.000\pm.000$ & $.000\pm.000$ \\
\bottomrule
\end{tabular}
}
\end{subtable}
\caption{Empirical results focusing on player B (rows) versus player A (columns) for games with $T=1000$, $K=3$, $N_A=15$, and $N_B=15$.}
\end{table}


\begin{table}[htb!p]
\begin{subtable}[h]{\linewidth}
\centering
\caption{Observable Expected Payoff Normalized Error}
\resizebox{\linewidth}{!}{
\begin{tabular}{rcccc}
\toprule
 & \ttsc{MARA} & \ttsc{CUCB-DRA} & \ttsc{Edge} & Random \\
\midrule
\ttsc{MARA} & $.047\pm.029$ & $.211\pm.108$ & $.100\pm.100$ & $.102\pm.102$ \\
\ttsc{CUCB-DRA} & $.255\pm.102$ & $.090\pm.088$ & $.096\pm.096$ & $.095\pm.096$ \\
\ttsc{Edge} & $.112\pm.058$ & $.132\pm.120$ & $.087\pm.087$ & $.089\pm.089$ \\
Random & $.124\pm.058$ & $.151\pm.117$ & $.088\pm.088$ & $.091\pm.091$ \\
\bottomrule
\end{tabular}
}
\end{subtable}

\bigskip

\begin{subtable}[h]{\linewidth}
\centering
\caption{Observable Max Payoff Normalized Error}
\label{subtab:non-zero max payoff 1}
\resizebox{\linewidth}{!}{
\begin{tabular}{rcccc}
\toprule
 & \ttsc{MARA} & \ttsc{CUCB-DRA} & \ttsc{Edge} & Random \\
\midrule
\ttsc{MARA} & $.008\pm.008$ & $.000\pm.000$ & $.000\pm.000$ & $.000\pm.000$ \\
\ttsc{CUCB-DRA} & $.008\pm.008$ & $.008\pm.008$ & $.000\pm.000$ & $.000\pm.000$ \\
\ttsc{Edge} & $.008\pm.008$ & $.000\pm.000$ & $.000\pm.000$ & $.011\pm.011$ \\
Random & $.008\pm.008$ & $.000\pm.000$ & $.000\pm.000$ & $.008\pm.008$ \\
\bottomrule
\end{tabular}
}
\end{subtable}

\bigskip

\begin{subtable}[h]{\linewidth}
\centering
\caption{Supremum Payoff Normalized Error}
\label{subtab:high supremum payoff A}
\resizebox{\linewidth}{!}{
\begin{tabular}{rcccc}
\toprule
 & \ttsc{MARA} & \ttsc{CUCB-DRA} & \ttsc{Edge} & Random \\
\midrule
\ttsc{MARA} & $.186\pm.124$ & $.182\pm.125$ & $.204\pm.118$ & $.209\pm.114$ \\
\ttsc{CUCB-DRA} & $.174\pm.125$ & $.213\pm.112$ & $.210\pm.114$ & $.209\pm.115$ \\
\ttsc{Edge} & $.082\pm.078$ & $.119\pm.105$ & $.135\pm.114$ & $.139\pm.115$ \\
Random & $.076\pm.072$ & $.109\pm.098$ & $.127\pm.110$ & $.127\pm.109$ \\
\bottomrule
\end{tabular}
}
\end{subtable}
\caption{Empirical results focusing on player A (rows) versus player B (columns) for games with $T=1000$, $K=5$, $N_A=15$, and $N_B=15$.}
\end{table}


\begin{table}[htb!p]
\begin{subtable}[h]{\linewidth}
\centering
\caption{Observable Expected Payoff Normalized Error}
\resizebox{\linewidth}{!}{
\begin{tabular}{rcccc}
\toprule
 & \ttsc{MARA} & \ttsc{CUCB-DRA} & \ttsc{Edge} & Random \\
\midrule
\ttsc{MARA} & $.041\pm.018$ & $.112\pm.082$ & $.076\pm.076$ & $.072\pm.072$ \\
\ttsc{CUCB-DRA} & $.039\pm.039$ & $.163\pm.084$ & $.075\pm.075$ & $.078\pm.078$ \\
\ttsc{Edge} & $.115\pm.052$ & $.168\pm.080$ & $.079\pm.079$ & $.077\pm.077$ \\
Random & $.111\pm.051$ & $.169\pm.082$ & $.076\pm.076$ & $.081\pm.081$ \\
\bottomrule
\end{tabular}
}
\end{subtable}

\bigskip

\begin{subtable}[h]{\linewidth}
\centering
\caption{Observable Max Payoff Normalized Error}
\resizebox{\linewidth}{!}{
\begin{tabular}{rcccc}
\toprule
 & \ttsc{MARA} & \ttsc{CUCB-DRA} & \ttsc{Edge} & Random \\
\midrule
\ttsc{MARA} & $.000\pm.000$ & $.000\pm.000$ & $.000\pm.000$ & $.000\pm.000$ \\
\ttsc{CUCB-DRA} & $.000\pm.000$ & $.000\pm.000$ & $.000\pm.000$ & $.000\pm.000$ \\
\ttsc{Edge} & $.000\pm.000$ & $.000\pm.000$ & $.000\pm.000$ & $.000\pm.000$ \\
Random & $.000\pm.000$ & $.000\pm.000$ & $.000\pm.000$ & $.000\pm.000$ \\
\bottomrule
\end{tabular}
}
\end{subtable}

\bigskip

\begin{subtable}[h]{\linewidth}
\centering
\caption{Supremum Payoff Normalized Error}
\resizebox{\linewidth}{!}{
\begin{tabular}{rcccc}
\toprule
 & \ttsc{MARA} & \ttsc{CUCB-DRA} & \ttsc{Edge} & Random \\
\midrule
\ttsc{MARA} & $.000\pm.000$ & $.000\pm.000$ & $.000\pm.000$ & $.000\pm.000$ \\
\ttsc{CUCB-DRA} & $.000\pm.000$ & $.000\pm.000$ & $.000\pm.000$ & $.000\pm.000$ \\
\ttsc{Edge} & $.000\pm.000$ & $.000\pm.000$ & $.000\pm.000$ & $.000\pm.000$ \\
Random & $.000\pm.000$ & $.000\pm.000$ & $.000\pm.000$ & $.000\pm.000$ \\
\bottomrule
\end{tabular}
}
\end{subtable}
\caption{Empirical results focusing on player B (rows) versus player A (columns) for games with $T=1000$, $K=5$, $N_A=15$, and $N_B=15$.}
\end{table}


\begin{table}[htb!p]
\begin{subtable}[h]{\linewidth}
\centering
\caption{Observable Expected Payoff Normalized Error}
\resizebox{\linewidth}{!}{
\begin{tabular}{rcccc}
\toprule
 & \ttsc{MARA} & \ttsc{CUCB-DRA} & \ttsc{Edge} & Random \\
\midrule
\ttsc{MARA} & $.028\pm.016$ & $.139\pm.084$ & $.059\pm.059$ & $.055\pm.055$ \\
\ttsc{CUCB-DRA} & $.198\pm.073$ & $.060\pm.058$ & $.060\pm.060$ & $.061\pm.061$ \\
\ttsc{Edge} & $.131\pm.056$ & $.095\pm.086$ & $.056\pm.056$ & $.057\pm.057$ \\
Random & $.126\pm.050$ & $.114\pm.081$ & $.054\pm.054$ & $.059\pm.059$ \\
\bottomrule
\end{tabular}
}
\end{subtable}

\bigskip

\begin{subtable}[h]{\linewidth}
\centering
\caption{Observable Max Payoff Normalized Error}
\resizebox{\linewidth}{!}{
\begin{tabular}{rcccc}
\toprule
 & \ttsc{MARA} & \ttsc{CUCB-DRA} & \ttsc{Edge} & Random \\
\midrule
\ttsc{MARA} & $.000\pm.000$ & $.000\pm.000$ & $.000\pm.000$ & $.000\pm.000$ \\
\ttsc{CUCB-DRA} & $.000\pm.000$ & $.000\pm.000$ & $.000\pm.000$ & $.000\pm.000$ \\
\ttsc{Edge} & $.000\pm.000$ & $.000\pm.000$ & $.000\pm.000$ & $.000\pm.000$ \\
Random & $.000\pm.000$ & $.000\pm.000$ & $.000\pm.000$ & $.000\pm.000$ \\
\bottomrule
\end{tabular}
}
\end{subtable}

\bigskip

\begin{subtable}[h]{\linewidth}
\centering
\caption{Supremum Payoff Normalized Error}
\resizebox{\linewidth}{!}{
\begin{tabular}{rcccc}
\toprule
 & \ttsc{MARA} & \ttsc{CUCB-DRA} & \ttsc{Edge} & Random \\
\midrule
\ttsc{MARA} & $.000\pm.000$ & $.000\pm.000$ & $.000\pm.000$ & $.000\pm.000$ \\
\ttsc{CUCB-DRA} & $.000\pm.000$ & $.000\pm.000$ & $.000\pm.000$ & $.000\pm.000$ \\
\ttsc{Edge} & $.000\pm.000$ & $.000\pm.000$ & $.000\pm.000$ & $.000\pm.000$ \\
Random & $.000\pm.000$ & $.000\pm.000$ & $.000\pm.000$ & $.000\pm.000$ \\
\bottomrule
\end{tabular}
}
\end{subtable}
\caption{Empirical results focusing on player A (rows) versus player B (columns) for games with $T=1000$, $K=5$, $N_A=20$, and $N_B=15$.}
\end{table}


\begin{table}[htb!p]
\begin{subtable}[h]{\linewidth}
\centering
\caption{Observable Expected Payoff Normalized Error}
\resizebox{\linewidth}{!}{
\begin{tabular}{rcccc}
\toprule
 & \ttsc{MARA} & \ttsc{CUCB-DRA} & \ttsc{Edge} & Random \\
\midrule
\ttsc{MARA} & $.073\pm.043$ & $.136\pm.106$ & $.116\pm.115$ & $.118\pm.118$ \\
\ttsc{CUCB-DRA} & $.058\pm.058$ & $.196\pm.091$ & $.113\pm.113$ & $.113\pm.113$ \\
\ttsc{Edge} & $.134\pm.069$ & $.203\pm.093$ & $.117\pm.117$ & $.119\pm.119$ \\
Random & $.135\pm.069$ & $.205\pm.092$ & $.114\pm.114$ & $.121\pm.121$ \\
\bottomrule
\end{tabular}
}
\end{subtable}

\bigskip

\begin{subtable}[h]{\linewidth}
\centering
\caption{Observable Max Payoff Normalized Error}
\label{subtab:non-zero max payoff 2}
\resizebox{\linewidth}{!}{
\begin{tabular}{rcccc}
\toprule
 & \ttsc{MARA} & \ttsc{CUCB-DRA} & \ttsc{Edge} & Random \\
\midrule
\ttsc{MARA} & $.008\pm.008$ & $.000\pm.000$ & $.000\pm.000$ & $.000\pm.000$ \\
\ttsc{CUCB-DRA} & $.008\pm.008$ & $.000\pm.000$ & $.000\pm.000$ & $.000\pm.000$ \\
\ttsc{Edge} & $.008\pm.008$ & $.000\pm.000$ & $.000\pm.000$ & $.000\pm.000$ \\
Random & $.008\pm.008$ & $.000\pm.000$ & $.000\pm.000$ & $.000\pm.000$ \\
\bottomrule
\end{tabular}
}
\end{subtable}

\bigskip

\begin{subtable}[h]{\linewidth}
\centering
\caption{Supremum Payoff Normalized Error}
\label{subtab:high supremum payoff B}
\resizebox{\linewidth}{!}{
\begin{tabular}{rcccc}
\toprule
 & \ttsc{MARA} & \ttsc{CUCB-DRA} & \ttsc{Edge} & Random \\
\midrule
\ttsc{MARA} & $.221\pm.104$ & $.016\pm.016$ & $.138\pm.115$ & $.130\pm.111$ \\
\ttsc{CUCB-DRA} & $.092\pm.086$ & $.031\pm.030$ & $.111\pm.099$ & $.129\pm.111$ \\
\ttsc{Edge} & $.033\pm.032$ & $.032\pm.031$ & $.103\pm.094$ & $.105\pm.095$ \\
Random & $.037\pm.037$ & $.032\pm.031$ & $.100\pm.092$ & $.110\pm.099$ \\
\bottomrule
\end{tabular}
}
\end{subtable}
\caption{Empirical results focusing on player B (rows) versus player A (columns) for games with $T=1000$, $K=5$, $N_A=20$, and $N_B=15$.}
\end{table}


\begin{table}[htb!p]
\begin{subtable}[h]{\linewidth}
\centering
\caption{Observable Expected Payoff Normalized Error}
\resizebox{\linewidth}{!}{
\begin{tabular}{rcccc}
\toprule
 & \ttsc{MARA} & \ttsc{CUCB-DRA} & \ttsc{Edge} & Random \\
\midrule
\ttsc{MARA} & $.084\pm.020$ & $.226\pm.112$ & $.100\pm.100$ & $.097\pm.097$ \\
\ttsc{CUCB-DRA} & $.259\pm.105$ & $.093\pm.088$ & $.097\pm.097$ & $.096\pm.096$ \\
\ttsc{Edge} & $.120\pm.054$ & $.146\pm.127$ & $.091\pm.091$ & $.093\pm.093$ \\
Random & $.116\pm.054$ & $.175\pm.103$ & $.089\pm.089$ & $.091\pm.091$ \\
\bottomrule
\end{tabular}
}
\end{subtable}

\bigskip

\begin{subtable}[h]{\linewidth}
\centering
\caption{Observable Max Payoff Normalized Error}
\resizebox{\linewidth}{!}{
\begin{tabular}{rcccc}
\toprule
 & \ttsc{MARA} & \ttsc{CUCB-DRA} & \ttsc{Edge} & Random \\
\midrule
\ttsc{MARA} & $.000\pm.000$ & $.000\pm.000$ & $.000\pm.000$ & $.000\pm.000$ \\
\ttsc{CUCB-DRA} & $.000\pm.000$ & $.000\pm.000$ & $.000\pm.000$ & $.000\pm.000$ \\
\ttsc{Edge} & $.000\pm.000$ & $.000\pm.000$ & $.000\pm.000$ & $.000\pm.000$ \\
Random & $.000\pm.000$ & $.000\pm.000$ & $.000\pm.000$ & $.000\pm.000$ \\
\bottomrule
\end{tabular}
}
\end{subtable}

\bigskip

\begin{subtable}[h]{\linewidth}
\centering
\caption{Supremum Payoff Normalized Error}
\resizebox{\linewidth}{!}{
\begin{tabular}{rcccc}
\toprule
 & \ttsc{MARA} & \ttsc{CUCB-DRA} & \ttsc{Edge} & Random \\
\midrule
\ttsc{MARA} & $.000\pm.000$ & $.000\pm.000$ & $.000\pm.000$ & $.000\pm.000$ \\
\ttsc{CUCB-DRA} & $.000\pm.000$ & $.000\pm.000$ & $.000\pm.000$ & $.000\pm.000$ \\
\ttsc{Edge} & $.000\pm.000$ & $.000\pm.000$ & $.000\pm.000$ & $.000\pm.000$ \\
Random & $.000\pm.000$ & $.000\pm.000$ & $.000\pm.000$ & $.000\pm.000$ \\
\bottomrule
\end{tabular}
}
\end{subtable}
\caption{Empirical results focusing on player A (rows) versus player B (columns) for games with $T=1000$, $K=5$, $N_A=20$, and $N_B=20$.}
\end{table}


\begin{table}[htb!p]
\begin{subtable}[h]{\linewidth}
\centering
\caption{Observable Expected Payoff Normalized Error}
\resizebox{\linewidth}{!}{
\begin{tabular}{rcccc}
\toprule
 & \ttsc{MARA} & \ttsc{CUCB-DRA} & \ttsc{Edge} & Random \\
\midrule
\ttsc{MARA} & $.034\pm.017$ & $.119\pm.086$ & $.080\pm.080$ & $.075\pm.075$ \\
\ttsc{CUCB-DRA} & $.039\pm.037$ & $.167\pm.087$ & $.079\pm.079$ & $.078\pm.078$ \\
\ttsc{Edge} & $.122\pm.054$ & $.175\pm.083$ & $.080\pm.080$ & $.082\pm.082$ \\
Random & $.115\pm.054$ & $.177\pm.084$ & $.080\pm.080$ & $.082\pm.082$ \\
\bottomrule
\end{tabular}
}
\end{subtable}

\bigskip

\begin{subtable}[h]{\linewidth}
\centering
\caption{Observable Max Payoff Normalized Error}
\resizebox{\linewidth}{!}{
\begin{tabular}{rcccc}
\toprule
 & \ttsc{MARA} & \ttsc{CUCB-DRA} & \ttsc{Edge} & Random \\
\midrule
\ttsc{MARA} & $.000\pm.000$ & $.000\pm.000$ & $.000\pm.000$ & $.000\pm.000$ \\
\ttsc{CUCB-DRA} & $.000\pm.000$ & $.000\pm.000$ & $.000\pm.000$ & $.000\pm.000$ \\
\ttsc{Edge} & $.000\pm.000$ & $.000\pm.000$ & $.000\pm.000$ & $.000\pm.000$ \\
Random & $.000\pm.000$ & $.000\pm.000$ & $.000\pm.000$ & $.000\pm.000$ \\
\bottomrule
\end{tabular}
}
\end{subtable}

\bigskip

\begin{subtable}[h]{\linewidth}
\centering
\caption{Supremum Payoff Normalized Error}
\resizebox{\linewidth}{!}{
\begin{tabular}{rcccc}
\toprule
 & \ttsc{MARA} & \ttsc{CUCB-DRA} & \ttsc{Edge} & Random \\
\midrule
\ttsc{MARA} & $.000\pm.000$ & $.000\pm.000$ & $.000\pm.000$ & $.000\pm.000$ \\
\ttsc{CUCB-DRA} & $.000\pm.000$ & $.000\pm.000$ & $.000\pm.000$ & $.000\pm.000$ \\
\ttsc{Edge} & $.000\pm.000$ & $.000\pm.000$ & $.000\pm.000$ & $.000\pm.000$ \\
Random & $.000\pm.000$ & $.000\pm.000$ & $.000\pm.000$ & $.000\pm.000$ \\
\bottomrule
\end{tabular}
}
\end{subtable}
\caption{Empirical results focusing on player B (rows) versus player A (columns) for games with $T=1000$, $K=5$, $N_A=20$, and $N_B=20$.}
\label{tab:last result table}
\end{table}

\end{document}